\documentclass[onecolumn,11pt]{IEEEtran}

\usepackage[utf8]{inputenc} % allow utf-8 input
\usepackage[T1]{fontenc}    % use 8-bit T1 fonts
\usepackage{hyperref}       % hyperlinks
\usepackage{url}            % simple URL typesetting
\usepackage{booktabs}       % professional-quality tables
\usepackage{amsfonts}       % blackboard math symbols
\usepackage{nicefrac}       % compact symbols for 1/2, etc.
\usepackage{microtype}      % microtypography

\usepackage{cite}
\usepackage{amsmath, amssymb, bm, epsfig, amsthm}
\usepackage{graphicx}
\usepackage{psfrag}
\usepackage{algorithm,algorithmic}
\usepackage{bbm}
\usepackage{etoolbox}
\usepackage{enumerate}
\usepackage{tikz}
\usetikzlibrary{shapes,arrows}
\usetikzlibrary{patterns}
\usepackage{pgfplots}
\usepackage{xspace}
\usepackage{enumerate}
\usepackage{mathtools}
\usepackage{epstopdf}

\newtheorem{definition}{Definition}
\newtheorem{theorem}{Theorem}

\def\beq{\begin{equation}}
\def\eeq{\end{equation}}

\def\R{{\mathbb{R}}}

\def\argmin{\mathop{\mathrm{arg\,min}}}
\def\argmax{\mathop{\mathrm{arg\,max}}}
\def\diag{\mathop{\mathrm{diag}}}
\def\Diag{\mathop{\mathrm{Diag}}}

\DeclareMathOperator{\Tr}{tr}

\DeclareMathOperator{\Cov}{Cov}
\def\x{\times}

\def\xhat{\widehat{x}}
\def\arr{\rightarrow}
\def\Exp{\mathbb{E}}

\def\km1{k\! - \! 1}
\def\kp1{k\! + \! 1}

\newcommand{\zero}{\mathbf{0}}

\newcommand{\dbf}{\mathbf{d}}

\newcommand{\fbf}{\mathbf{f}}
\newcommand{\gbf}{\mathbf{g}}

\newcommand{\pbf}{\mathbf{p}}

\newcommand{\qbf}{\mathbf{q}}

\newcommand{\rbf}{\mathbf{r}}

\newcommand{\sbf}{\mathbf{s}}

\newcommand{\ubf}{\mathbf{u}}

\newcommand{\vbf}{\mathbf{v}}

\newcommand{\wbf}{\mathbf{w}}

\newcommand{\xbf}{\mathbf{x}}
\newcommand{\xbfhat}{\widehat{\mathbf{x}}}

\newcommand{\ybf}{\mathbf{y}}
\newcommand{\zbf}{\mathbf{z}}

\newcommand{\Abf}{\mathbf{A}}

\newcommand{\Cbar}{\overline{C}}

\newcommand{\Hbf}{\mathbf{H}}

\newcommand{\Ibf}{\mathbf{I}}

\newcommand{\Pbf}{\mathbf{P}}

\newcommand{\Qbf}{\mathbf{Q}}

\newcommand{\Sbf}{\mathbf{S}}
\newcommand{\Ubf}{\mathbf{U}}
\newcommand{\Vbf}{\mathbf{V}}

\newcommand{\Xhat}{\widehat{X}}

\newcommand{\gammahat}{{\widehat{\gamma}}}

\def\xibf{{\boldsymbol \xi}}

\def\taubar{{\overline{\tau}}}

\def\tauhat{{\widehat{\tau}}}

\newcommand{\lambdabar}{{\overline{\lambda}}}
\newcommand{\thetabar}{{\overline{\theta}}}
\newcommand{\thetabf}{{\boldsymbol{\theta}}}
\newcommand{\thetabfbar}{{\overline{\boldsymbol{\theta}}}}
\newcommand{\thetahat}{{\widehat{\theta}}}
\newcommand{\thetabfhat}{{\widehat{\boldsymbol{\theta}}}}
\newcommand{\phibf}{{\boldsymbol{\phi}}}
\newcommand{\mubar}{\overline{\mu}}

\def\alphabar{\overline{\alpha}}
\def\etabar{\overline{\eta}}
\def\gammabar{\overline{\gamma}}

\def\Ecal{{\mathcal E}}

\newcommand{\indic}[1]{\mathbbm{1}_{ \{ {#1} \} }}

\newcommand{\tran}{^{\text{\sf T}}}

\newcommand{\bkt}[1]{{\langle #1 \rangle}}

\def\PLeq{\stackrel{PL(2)}{=}}
\def\Norm{{\mathcal N}}

\title{Rigorous Dynamics and Consistent Estimation in Arbitrarily Conditioned Linear Systems}

\author{
  Alyson K. Fletcher, Mojtaba Sahraee-Ardakan, Philip Schniter, and Sundeep Rangan%
  \thanks{A. K. Fletcher and M. Sahraee-Ardakan are with the University of California, Los Angeles.}
  \thanks{P. Schniter is with The Ohio State University.}
  \thanks{S. Rangan is with New York University.}
}

\begin{document}

\maketitle

\begin{abstract}
The problem of estimating a random vector $\xbf$
from noisy linear measurements $\ybf=\Abf\xbf+\wbf$ with unknown parameters
on the distributions of $\xbf$ and $\wbf$, which must also be learned,
arises in a wide range of statistical learning and linear
inverse problems.  We show that a computationally simple
iterative message-passing algorithm can provably obtain asymptotically
consistent estimates in a certain high-dimensional large-system
limit (LSL) under very general parameterizations.
Previous message passing techniques have required i.i.d.\
sub-Gaussian $\Abf$ matrices and often fail when the matrix is
ill-conditioned. The proposed algorithm, called
adaptive vector approximate message passing (Adaptive VAMP) with auto-tuning,
applies to all right-rotationally random $\Abf$.
Importantly, this class includes matrices with arbitrarily bad conditioning.
We show that the parameter estimates and mean squared error (MSE) of $\xbf$
in each iteration converge to deterministic limits
that can be precisely predicted by a simple set of
state evolution (SE) equations.
In addition, a simple testable condition is provided in which
the MSE matches the Bayes-optimal value
predicted by the replica method.
The paper thus provides a computationally simple method with
provable guarantees of optimality and consistency over a
large class of linear inverse problems.
\end{abstract}

%\begin{keywords}
%Message passing, approximate inference, sparse linear regression, approximate message passing.
%\end{keywords}

\section{Introduction}

Consider the problem of estimating a random vector $\xbf^0$
from linear measurements $\ybf$ of the form
\beq \label{eq:yAx}
    \ybf = \Abf \xbf^0 + \wbf, \quad
    \wbf \sim {\mathcal N}(\zero,\theta_2^{-1} \Ibf),
    \quad
    \xbf^0 \sim p(\xbf|\thetabf_1),
\eeq
where $\Abf\in\R^{M\times N}$ is a known matrix,
$p(\xbf|\thetabf_1)$ is a density on $\xbf^0$ with
parameters $\thetabf_1$,
$\wbf$ is additive white Gaussian noise (AWGN) independent of $\xbf^0$,
and $\theta_2 > 0$ is the noise precision (inverse variance).
The goal is to estimate $\xbf^0$,
while simultaneously learning the unknown parameters
$\thetabf := (\thetabf_1,\theta_2)$, from the data $\ybf$ and $\Abf$.
This problem arises in Bayesian forms
of linear inverse problems in signal processing,
as well as in linear regression in statistics.

Exact estimation of the parameters $\thetabf$ via maximum likelihood or other methods is generally intractable.
One promising class of approximate methods combines
approximate message passing (AMP) \cite{DonohoMM:09} with expectation-maximization (EM)\@.
AMP and its generalizations (e.g., \cite{Rangan:11-ISIT})
constitute a powerful, relatively new class of algorithms based on expectation propagation \cite{Minka:01} (EP)-type techniques.
The resulting AMP algorithms are computationally fast
and have been successfully applied to a wide range of problems,
e.g., \cite{FletcherRVB:11,Schniter:11,SomS:12,ZinielS:13b,Vila:TSP:14,Schniter:TSP:15,Ziniel:TSP:15,fletcher2014scalable}.
Most importantly, for large, zero-mean, sub-Gaussian i.i.d.\ random matrices $\Abf$,
the performance of these AMP methods can be exactly predicted by a scalar \emph{state evolution}
(SE) \cite{BayatiM:11,javanmard2013state} that provides testable conditions
for optimality, even for non-convex priors.  When the parameters $\thetabf$ are unknown,
AMP can be easily combined with EM for joint learning of the parameters $\thetabf$ and
vector $\xbf$~\cite{krzakala2012statistical,vila2013expectation,KamRanFU:12-IT}.
Under more general conditions on $\Abf$, however, not only do the theoretical results
not hold, but standard AMP techniques often diverge and require a variety of modifications for stability
\cite{RanSchFle:14-ISIT,Vila:ICASSP:15,manoel2015swamp,rangan2017inference}.
For example, when $\Abf$ has nonzero mean, its largest singular value grows with the
problem size, making $\Abf$ arbitrarily poorly conditioned.

A recent work \cite{fletcher2016emvamp} combined EM with the so-called Vector AMP (VAMP)
method of \cite{rangan2016vamp}. Similar to AMP, VAMP is based on
EP-like \cite{Minka:01} approximations of belief propagation \cite{rangan2016vamp} and can also be considered as a special case
 of expectation consistent (EC) approximate inference \cite{opper2004expectation,OppWin:05,fletcher2016expectation}.
VAMP's key attraction is that it applies to a larger class of matrices $\Abf$ than
the original AMP method.
It
has provable SE analyses and convergence guarantees that apply to all
right-rotationally invariant matrices $\Abf$ \cite{rangan2016vamp,takeuchi2017rigorous}~--
a significantly larger class of matrices than i.i.d.\ sub-Gaussian.
Under further mild conditions, the mean-squared error (MSE)
of VAMP matches the replica prediction for
optimality~\cite{tulino2013support,barbier2016mutual,reeves2016replica}.
When the distributions on $\xbf$ and $\wbf$ are unknown,
the work \cite{fletcher2016emvamp} proposed to combine EM and VAMP using the approximate inference
framework of \cite{heskes2004approximate}.
While \cite{fletcher2016emvamp} provided
numerical simulations suggesting excellent performance for EM-VAMP on several
synthetic problems, there were no provable convergence guarantees.

The contributions of this work are as follows:
\begin{itemize}
\item \emph{Rigorous state evolution analysis:}  We provide a rigorous analysis of
a generalization of EM-VAMP that we call Adaptive VAMP\@.
Similar to the analysis of VAMP, we consider a certain large-system limit (LSL)
where the matrix $\Abf$ is random and right-rotationally invariant.
Importantly, this class of matrices includes much more than i.i.d.\ Gaussian, as used in the
original LSL analysis of
 Bayati and Montanari \cite{BayatiM:11}.
It is shown (Theorem~\ref{thm:em-se}) that,
in this LSL, the parameter estimates at each iteration converge to deterministic limits $\thetabfbar_k$
that can be computed from a set of SE equations that extend those of VAMP\@.
The analysis also exactly characterizes the asymptotic joint distribution of the
estimates $\xbfhat$ and the true vector $\xbf^0$.
The SE equations depend only on the initial parameter estimate
, the adaptation function (to be discussed), and statistics on the matrix $\Abf$,
the vector $\xbf^0$, and the noise $\wbf$.

\item \emph{Asymptotic consistency}:
It is also shown (Theorem~\ref{thm:theta1cons}) that, under an additional identifiability condition
and a simple auto-tuning procedure,
Adaptive VAMP can yield provably consistent parameter estimates in the LSL\@.
This approach is inspired by an ML-estimation approach from \cite{KamRanFU:12-IT}.
Remarkably, the result is true under very general problem formulations.

\item \emph{Bayes optimality}:  In the case when the parameter estimates converge to the
true value, the behavior of adaptive VAMP matches that of VAMP\@.  In this case,
it is shown in \cite{rangan2016vamp} that, when the SE equations have a unique fixed point,
the MSE of VAMP matches the MSE of the Bayes optimal estimator, as predicted by the replica
method~\cite{tulino2013support,barbier2016mutual,reeves2016replica}.
\end{itemize}

In this way, we have developed a computationally efficient inference scheme for
a large class of linear inverse problems.
In a certain high-dimensional limit, our scheme guarantees that
(i) the performance of the algorithm can be exactly characterized,
(ii) the parameter estimates $\thetabfhat$ are asymptotically consistent,
and
(iii) the algorithm has testable conditions for when
the signal estimates $\xbfhat$ match the replica prediction of Bayes optimality.

\section{VAMP with Adaptation}

We assume that the prior density on $\xbf$ can be written as
\beq \label{eq:pxf}
    p(\xbf|\thetabf_1) = \frac{1}{Z_1(\thetabf_1)}\exp\left[ -f_1(\xbf|\thetabf_1) \right],
    \quad  f_1(\xbf|\thetabf_1) = \sum_{n=1}^N f_{1}(x_n|\thetabf_1),
\eeq
where $f_1(\cdot)$ is a separable penalty function, $\thetabf_1$ is a parameter vector and
$Z_1(\thetabf_1)$ is a normalization constant.  With some abuse of notation,
we have used $f_1(\cdot)$ for the function on the vector $\xbf$ and its components $x_n$.
Since $f_1(\xbf|\thetabf_1)$ is separable,  $\xbf$ has i.i.d.\ components
conditioned on $\thetabf_1$.
The likelihood function under the AWGN model \eqref{eq:yAx} can be written as
\beq \label{eq:pyxf}
    p(\ybf|\xbf,\theta_2)
    := \frac{1}{Z_2(\theta_2)}
    \exp\left[-f_2(\xbf,\ybf|\theta_2)\right], \quad
    f_2(\xbf,\ybf|\theta_2)
    := \frac{\theta_2}{2}\|\ybf-\Abf\xbf\|^2,
\eeq
where $Z_2(\theta_2) = (2\pi/\theta_2)^{N/2}$.
The joint density of $\xbf,\ybf$ given parameters
$\thetabf=(\thetabf_1,\theta_2)$ is then
\beq \label{eq:pxy}
    p(\xbf,\ybf|\thetabf) = p(\xbf|\thetabf_1)p(\ybf|\xbf,\theta_2).
\eeq
The problem is to estimate the parameters $\thetabf$ along with
the vector $\xbf^0$.

%%%%%%%%%%%%%%%%%%%%%%%%%%%%%%%%%%%
\begin{algorithm}[t]
\caption{Adaptive VAMP}
\begin{algorithmic}[1]  \label{algo:em-vamp}
\REQUIRE{Matrix $\Abf\in\R^{M\times N}$, measurement vector $\ybf$,
denoiser function $\gbf_1(\cdot)$, statistic function $\phi_1(\cdot)$,
adaptation function $T_1(\cdot)$ and number of iterations $N_{\rm it}$.  }
\STATE{ Select initial $\rbf_{10}$, $\gamma_{10}\geq 0$, $\thetabfhat_{10}$, $\thetahat_{20}$.}
\FOR{$k=0,1,\dots,N_{\rm it}-1$}

    \STATE{// Input denoising }
    \STATE{$\xbfhat_{1k} = \gbf_1(\rbf_{1k},\gamma_{1k},\thetabfhat_{1k}))$,
        \qquad
        $\eta_{1k}^{-1} = \gamma_{1k}/\bkt{ \gbf_1'(\rbf_{1k},\gamma_{1k},\thetabfhat_{1k})}$}
        \label{line:x1}
    \STATE{$\gamma_{2k} = \eta_{1k} - \gamma_{1k}$}
        \label{line:gam2}
    \STATE{$\rbf_{2k} = (\eta_{1k}\xbfhat_{1k} - \gamma_{1k}\rbf_{1k})/\gamma_{2k}$}
        \label{line:r2}
    \STATE{ }

    \STATE{// Input parameter update }
    \STATE{ $\thetabfhat_{1,\kp1} = T_1(\mu_{1k})$,
        \qquad
        $\mu_{1k} = \bkt{\phi_1(\rbf_{1k},\gamma_{1k},\thetabfhat_{1k})}$ }
        \label{line:theta1}
    \STATE{}

    \STATE{// Output estimation }
    \STATE{$\xbfhat_{2k} = \Qbf_k^{-1}( \thetahat_{2k}\Abf\tran\ybf + \gamma_{2k}\rbf_{2k})$,
        \qquad
        $\Qbf_k = \thetahat_{2k}\Abf\tran\Abf + \gamma_{2k}\Ibf$}
        \label{line:x2}
    \STATE{$\eta_{2k}^{-1} = (1/N)\Tr(\Qbf_k^{-1})$}
        \label{line:eta2}
    \STATE{$\gamma_{1,\kp1} = \eta_{2k} - \gamma_{2k}$}
        \label{line:gam1}
    \STATE{$\rbf_{1,\kp1} = (\eta_{2k}\xbfhat_{2k} - \gamma_{2k}\rbf_{2k})/\gamma_{1,\kp1}$}
        \label{line:r1}
    \STATE{}

    \STATE{// Output parameter update }
    \STATE{ $\thetahat_{2,\kp1}^{-1} =
        (1/N)\{\|\ybf-\Abf\xbfhat_{2k}\|^2 + \Tr(\Abf\Qbf_k^{-1}\Abf\tran )\}$ }
        \label{line:theta2}
\ENDFOR
\end{algorithmic}
\end{algorithm}
%%%%%%%%%%%%%%%%%%%%%%%%%%%%%%%%%%%

The steps of the proposed adaptive VAMP algorithm
that performs this estimation are shown in Algorithm~\ref{algo:em-vamp}.
Adaptive VAMP is a generalization of the EM-VAMP method in \cite{fletcher2016emvamp}.
At each iteration $k$, the algorithm produces, for $i=1,2$,
estimates $\thetabfhat_{ik}$ of the parameter $\thetabf_i$, along with estimates $\xbfhat_{ik}$
of the vector $\xbf^0$.  The algorithm is tuned by selecting three key functions:
(i) a \emph{denoiser function} $\gbf_1(\cdot)$;
(ii) an \emph{adaptation statistic} $\phi_1(\cdot)$; and
(iii) a \emph{parameter selection function} $T_1(\cdot)$.
The denoiser is used to produce the estimates $\xbfhat_{1k}$, while the adaptation statistic
and parameter estimation functions produce the estimates $\thetabfhat_{1k}$.

\paragraph*{Denoiser function}
The denoiser function $\gbf_1(\cdot)$ is discussed in detail in \cite{rangan2016vamp}
and is generally based on the prior $p(\xbf|\thetabf_1)$.  In the original EM-VAMP
algorithm~\cite{fletcher2016emvamp},
$\gbf_1(\cdot)$ is selected as the so-called minimum mean-squared error (MMSE) denoiser.
Specifically, in each iteration, the variables $\rbf_i$, $\gamma_i$ and $\thetabfhat_i$
were used to construct \emph{belief estimates},
\beq \label{eq:bidef}
    b_i(\xbf|\rbf_i,\gamma_i,\thetabfhat_i) \propto \exp\left[ -f_i(\xbf,\ybf|\thetabfhat_i) -
        \frac{\gamma_i}{2}\|\xbf-\rbf_i\|^2 \right],
\eeq
which represent estimates of the posterior density $p(\xbf|\ybf,\thetabf)$.
To keep the notation symmetric, we have written
$f_1(\xbf,\ybf|\thetabfhat_1)$ for $f_1(\xbf|\thetabfhat_1)$ even though the first penalty function
does not depend on $\ybf$.  The EM-VAMP method then selects $\gbf_1(\cdot)$ to be the
mean of the belief estimate,
\beq \label{eq:gmmse}
    \gbf_1(\rbf_1,\gamma_1,\thetabf_1) := \Exp\left[ \xbf |\rbf_1,\gamma_1,\thetabf_1 \right].
\eeq
For line~\ref{line:x1} of Algorithm~\ref{algo:em-vamp}, we define
$[\gbf_1'(\rbf_{1k},\gamma_{1k},\thetabf_1)]_n
:= \partial [\gbf_{1}(\rbf_{1k},\gamma_{1k},\thetabf_1)]_n/\partial r_{1n}$
and we use $\bkt{\cdot}$ for the empirical mean of a vector,
i.e., $\bkt{\ubf} = (1/N)\sum_{n=1}^N u_n$.
Hence, $\eta_{1k}$ in line~\ref{line:x1} is a scaled inverse divergence.
It is shown in \cite{rangan2016vamp} that, for the MMSE denoiser
\eqref{eq:gmmse}, $\eta_{1k}$ is the inverse average posterior variance.

\paragraph{Estimation for $\thetabf_1$ with finite statistics}
For the EM-VAMP algorithm~\cite{fletcher2016emvamp},
the parameter update for $\thetabfhat_{1,\kp1}$ is performed via a maximization
\beq \label{eq:thetaEM}
    \thetabfhat_{1,\kp1} = \argmax_{\thetabf_1} \Exp\left[ \ln p(\xbf|\thetabf_1)
        \left| \rbf_{1k},\gamma_{1k},\thetabfhat_{1k} \right. \right],
\eeq
where the expectation is with respect to the belief estimate $b_i(\cdot)$ in
\eqref{eq:bidef}.  It is shown in~\cite{fletcher2016emvamp} that using \eqref{eq:thetaEM}
is equivalent to approximating the M-step in the standard EM method.
In the adaptive VAMP method in Algorithm~\ref{algo:em-vamp}, the M-step maximization
\eqref{eq:thetaEM} is replaced by line~\ref{line:theta1}.  Note that line~\ref{line:theta1}
again uses $\bkt{\cdot}$ to denote empirical average,
\beq
    \mu_{1k} = \bkt{\phibf_1(\rbf_{1k},\gamma_{1k},\thetabfhat_{1k})}
        := \frac{1}{N} \sum_{n=1}^N \phi_1(r_{1k,n},\gamma_{1k},\thetabfhat_{1k}) \in \R^d,
\eeq
so $\mu_{1k}$ is the empirical average of some $d$-dimensional
statistic $\phi_1(\cdot)$ over the components
of $\rbf_{1k}$.  The parameter estimate update $\thetabfhat_{1,\kp1}$ is then computed from some
function of this statistic, $T_1(\mu_{1k})$.

We show in Appendix~\ref{sec:finite} that there are two important cases where the EM update
\eqref{eq:thetaEM} can be computed from a finite-dimensional statistic as in line~\ref{line:theta1}:
(i) The prior $p(\xbf|\thetabf_1)$ is given by an exponential family,
$f_1(\xbf|\thetabf_1) = \thetabf_1\tran \varphi(\xbf)$ for some sufficient statistic $\varphi(\xbf)$;
and (ii) There are a finite number of values for the parameter $\thetabf_1$.
For other cases, we can approximate more general parameterizations via discretization of the
parameter values $\thetabf_1$.  The updates in line~\ref{line:theta1} can also incorporate other
types of updates as we will see below.
But, we stress that it is preferable to compute the estimate for $\thetabf_1$ directly from the maximization
\eqref{eq:thetaEM}~-- the use of a finite-dimensional statistic is for the sake of analysis.

\paragraph{Estimation for $\theta_2$ with finite statistics}
It will be useful to also write the adaptation of $\theta_2$
in line~\ref{line:theta2} of Algorithm~\ref{algo:em-vamp}
in a similar form as line~\ref{line:theta1}.
First, take a singular value decomposition (SVD) of $\Abf$ of the form
\beq \label{eq:ASVD}
    \Abf=\Ubf\Sbf\Vbf\tran, \quad \Sbf = \Diag(\sbf),
\eeq
and define the transformed error and transformed noise,
\beq \label{eq:qerrdef}
    \qbf_k := \Vbf\tran(\rbf_{2k}-\xbf^0), \quad \xibf := \Ubf\tran\wbf.
\eeq
Then, it is shown in Appendix~\ref{sec:theta2TPf} that $\thetahat_{2,\kp1}$ in line~\ref{line:theta2}
can be written as
\beq \label{eq:theta2T}
    \thetahat_{2,\kp1} = T_2(\mu_{2k}) := \frac{1}{\mu_{2k}},
    \quad \mu_{2k} = \bkt{\phibf_2(\qbf_2,\xibf,\sbf,\gamma_{2k},\thetahat_{2k})}
\eeq
where
\beq \label{eq:phi2def}
    \phi_2(q,\xi,s,\gamma_2,\thetahat_2) := \frac{\gamma_2^2}{(s^2\thetahat_2+\gamma_2)^2}(sq + \xi)^2
        + \frac{s^2}{s^2\thetahat_2+\gamma_2}.
\eeq
Of course, we cannot directly compute $\qbf_k$ in \eqref{eq:qerrdef} since we do not know
the true $\xbf^0$.  Nevertheless, this form will be useful for analysis.

\section{State Evolution in the Large System Limit}

\subsection{Large System Limit} \label{sec:lsl}

Similar to the analysis of VAMP in \cite{rangan2016vamp},
we analyze Algorithm~\ref{algo:em-vamp} in a certain large-system limit (LSL).
The LSL framework was developed by Bayati and Montanari in \cite{BayatiM:11} and
we review some of the key definitions in Appendix~\ref{sec:empConv}.
As in the analysis of VAMP, the LSL considers a sequence of problems indexed by the vector
dimension $N$.  For each $N$, we assume that there is a ``true'' vector $\xbf^0\in\R^N$
that is observed through measurements of the form
\beq \label{eq:yAxslr}
    \ybf = \Abf\xbf^0 + \wbf \in \R^N, \quad \wbf \sim \Norm(\mathbf{0}, \theta_2^{-1}\Ibf_N),
\eeq
where $\Abf\in\R^{N\times N}$ is a known transform, $\wbf$ is white Gaussian noise with ``true'' precision $\theta_2$.  The noise precision $\theta_2$ does not change with $N$.

Identical to \cite{rangan2016vamp}, the transform $\Abf$
is modeled as a large, \emph{right-orthogonally invariant} random matrix.
Specifically, we assume that it has an SVD of the form \eqref{eq:ASVD},
where $\Ubf$ and $\Vbf$ are $N\times N$ orthogonal matrices
such that $\Ubf$ is deterministic and
$\Vbf$ is Haar distributed (i.e.\ uniformly distributed on the set of orthogonal matrices).
As described in \cite{rangan2016vamp},
although we have assumed a square matrix $\Abf$, we can consider general rectangular $\Abf$
by adding zero singular values.

Using the definitions in Appendix~\ref{sec:empConv},
we assume that the components of the singular-value vector $\sbf\in\R^N$ in \eqref{eq:ASVD}
converge empirically with second-order moments as
\beq \label{eq:Slim}
    \lim_{N \arr \infty} \{ s_n \} \PLeq S,
\eeq
for some non-negative random variable $S$ with $\Exp[S] > 0$ and $S \in [0,S_{\rm max}]$
for some finite maximum value $S_{\rm max}$.
Additionally, we assume that
the components of the true vector, $\xbf^0$, and the initial input to the denoiser, $\rbf_{10}$,
converge empirically as
\beq \label{eq:RX0lim}
    \lim_{N \arr \infty} \{ (r_{10,n}, x^0_n) \} \PLeq (R_{10},X^0), \quad
    R_{10} = X^0 + P_0, \quad P_0 \sim \Norm(0,\tau_{10}),
\eeq
where $X^0$ is a random variable representing the \emph{true
distribution} of the components $\xbf^0$; $P_0$ is an initial error and $\tau_{10}$ is an initial error variance.
The variable $X^0$ may be distributed as $X^0 \sim p(\cdot|\thetabf_1)$ for some
true parameter $\thetabf_1$.  However, in order to incorporate under-modeling, the existence of such a true
parameter is not required.
We also assume that the initial second-order term and parameter estimate converge almost surely as
\beq \label{eq:gam10lim}
    \lim_{N \arr \infty} (\gamma_{10},\thetabfhat_{10},\thetahat_{20})
        = (\gammabar_{10},\thetabfbar_{10},\thetabar_{20})
\eeq
for some $\gammabar_{10} > 0$ and $(\thetabfbar_{10},\thetabar_{20})$.

\subsection{Error and Sensitivity Functions}
We next need to introduce parametric forms of
two key terms from \cite{rangan2016vamp}:  error functions and sensitivity functions.
The error functions describe MSE
of the denoiser and output estimators under AWGN measurements.
Specifically, for the denoiser $g_1(\cdot,\gamma_1,\thetabfhat_1)$, we define the error function as
\begin{align}
    \Ecal_1(\gamma_1,\tau_1,\thetabfhat_1)
    := \Exp\left[ (g_1(R_1,\gamma_1,\thetabfhat_1)-X^0)^2 \right], \quad
     R_1 = X^0 + P, \quad P \sim \Norm(0,\tau_1), \label{eq:eps1}
\end{align}
where $X^0$ is distributed according to the true distribution of the components $\xbf^0$ (see above).
The function $\Ecal_1(\gamma_1,\tau_1,\thetabfhat_1)$ thus represents the MSE of the
estimate $\Xhat = g_1(R_1,\gamma_1,\thetabfhat_1)$ from a measurement $R_1$
corrupted by Gaussian noise of variance $\tau_1$ under the parameter estimate $\thetabfhat_1$.
For the output estimator, we define the error function as
\begin{align}
    \MoveEqLeft \Ecal_2(\gamma_2,\tau_2,\thetahat_2)
    := \lim_{N \arr \infty}
        \frac{1}{N} \Exp \| \gbf_2(\rbf_2,\gamma_2,\thetahat_2) -\xbf^0 \|^2, \nonumber \\
    & \xbf^0 = \rbf_2 + \qbf, \quad \qbf \sim \Norm(0,\tau_2 \Ibf), \quad
    \ybf = \Abf\xbf^0 + \wbf, \quad \wbf \sim \Norm(0,\theta_2^{-1} \Ibf),
    \label{eq:eps2}
\end{align}
which is the average per component error of the vector estimate under Gaussian noise.
The dependence on the true noise precision, $\theta_2$, is suppressed.

The sensitivity functions describe the expected divergence of the estimator.
For the denoiser, the sensitivity function is defined as
\begin{align}
    A_1(\gamma_1,\tau_1,\thetabfhat_1)
    := \Exp\left[ g_1'(R_1,\gamma_1,\thetabfhat_1) \right], \quad
     R_1 = X^0 + P, \quad P \sim \Norm(0,\tau_1), \label{eq:sens1}
\end{align}
which is the average derivative under a Gaussian noise input.  For the
output estimator, the sensitivity is defined as
\begin{align}
     A_2(\gamma_2,\tau_2,\thetahat_2)
    := \lim_{N \arr \infty}
        \frac{1}{N} \Tr\left[ \frac{\partial \gbf_2(\rbf_2,\gamma_2,\thetahat_2)}{\partial \rbf_2}
            \right],
\end{align}
where $\rbf_2$ is distributed as in \eqref{eq:eps2}.
The paper \cite{rangan2016vamp} discusses the error and sensitivity functions in detail
and shows how these functions can be easily evaluated.

\subsection{State Evolution Equations}
We can now describe our main result, which are the SE equations for Adaptive VAMP\@.
The equations are an extension of those in the VAMP paper \cite{rangan2016vamp},
with modifications for the parameter estimation.
For a given iteration $k \geq 1$, consider the set of components,
\[
     \{ (\xhat_{1k,n},r_{1k,n},x^0_n), ~ n=1,\ldots,N \}.
\]
This set represents the components of the true vector $\xbf^0$,
its corresponding estimate $\xbfhat_{1k}$ and the denoiser input
$\rbf_{1k}$.  We will show that, under certain assumptions,
these components converge empirically as
\beq \label{eq:limrx1}
    \lim_{N \arr \infty} \{ (\xhat_{1k,n},r_{1k,n},x^0_n) \}
    \PLeq (\Xhat_{1k},R_{1k},X^0),
\eeq
where the random variables $(\Xhat_{1k},R_{1k},X^0)$ are given by
\beq \label{eq:RX0var}
    R_{1k} = X^0 + P_k, \quad P_k \sim \Norm(0,\tau_{1k}), \quad
    \Xhat_{1k} = g_1(R_{1k},\gammabar_{1k},\thetabfbar_{1k}),
\eeq
for constants $\gammabar_{1k}$, $\thetabfbar_{1k}$ and $\tau_{1k}$ that will be defined below.
We will also see that $\thetabfhat_{1k} \arr \thetabfbar_{1k}$, so $\thetabfbar_{1k}$ represents the
asymptotic parameter estimate.
The model \eqref{eq:RX0var} shows that each component $r_{1k,n}$ appears as the true component $x^0_n$ plus
Gaussian noise.  The corresponding estimate $\xhat_{1k,n}$ then appears as the
denoiser output with $r_{1k,n}$ as the input and $\thetabfbar_{1k}$ as the parameter estimate.
Hence, the asymptotic behavior
of any component $x^0_n$ and its corresponding $\xhat_{1k,n}$ is identical to
a simple scalar system.  We will refer to \eqref{eq:limrx1}-\eqref{eq:RX0var} as the denoiser's \emph{scalar equivalent model}.

We will also show that these transformed errors $\qbf_k$ and noise $\xibf$ in \eqref{eq:qerrdef}
and singular values $\sbf$ converge
empirically to a set of independent random variables $(Q_k,\Xi,S)$ given by
\beq \label{eq:limqxi}
    \lim_{N \arr \infty} \{ (q_{k,n},\xi_n,s_n) \}
    \PLeq (Q_k,\Xi,S), \quad
    Q_k \sim \Norm(0,\tau_{2k}), \quad \Xi \sim \Norm(0,\theta_2^{-1}),
\eeq
where $S$ has the distribution of the singular values of $\Abf$,
$\tau_{2k}$ is a variance that will be defined below and $\theta_2$
is the true noise precision in the measurement model \eqref{eq:yAxslr}.
All the variables in \eqref{eq:limqxi} are independent.
Thus \eqref{eq:limqxi} is a scalar equivalent model for the output estimator.

The variance terms are defined recursively through the \emph{state evolution}
equations,
\begin{subequations} \label{eq:se}
\begin{align}
    \alphabar_{1k} &= A_1(\gammabar_{1k},\tau_{1k},\thetabfbar_{1k}), \quad
    \etabar_{1k} = \frac{\gammabar_{1k}}{\alphabar_{1k}}, \quad
    \gammabar_{2k} = \etabar_{1k} - \gammabar_{1k} \label{eq:eta1se} \\\
    \thetabfbar_{1,\kp1} &= T_1(\mubar_{1k}), \quad \mubar_{1k} =
        \Exp\left[ \phi_1(R_{1k},\gammabar_{1k},\thetabfbar_{1k}) \right] \label{eq:theta1se} \\
    \tau_{2k} &= \frac{1}{(1-\alphabar_{1k})^2}\left[
        \Ecal_1(\gammabar_{1k},\tau_{1k},\thetabfbar_{1k}) - \alphabar_{1k}^2\tau_{1k} \right],
            \label{eq:tau2se} \\
    \alphabar_{2k} &= A_2(\gammabar_{2k},\tau_{2k},\thetabar_{2k}), \quad
    \etabar_{2k} = \frac{\gammabar_{2k}}{\alphabar_{2k}}, \quad
    \gammabar_{1,\kp1} = \etabar_{2k} - \gammabar_{2k} \label{eq:eta2se} \\
    \thetabar_{2,\kp1} &= T_2(\mubar_{2k}), \quad \mubar_{2k} =
        \Exp\left[ \phi_2(Q_{k},\Xi,S,\gammabar_{2k},\thetabar_{2k}) \right] \label{eq:theta2se} \\
    \tau_{1,\kp1} &= \frac{1}{(1-\alphabar_{2k})^2}\left[
        \Ecal_2(\gammabar_{2k},\tau_{2k}) - \alphabar_{2k}^2\tau_{2k} \right],
            \label{eq:tau1se}
\end{align}
\end{subequations}
which are initialized with $\tau_{10} = \Exp[(R_{10}-X^0)^2]$ and the
$(\gammabar_{10},\thetabfbar_{10},\thetabar_{20})$
defined from the limit \eqref{eq:gam10lim}.  The expectation in \eqref{eq:theta1se} is
with respect to the random variables \eqref{eq:limrx1} and the expectation in \eqref{eq:theta2se} is
with respect to the random variables \eqref{eq:limqxi}.

\begin{theorem} \label{thm:em-se}
Consider the outputs of Algorithm~\ref{algo:em-vamp}.
Under the above assumptions and definitions, assume additionally that for all iterations $k$:
\begin{enumerate}[(i)]
\item The solution $\alphabar_{1k}$ from the SE equations \eqref{eq:se} satisfies
$\alphabar_{1k} \in (0,1)$.

\item The functions $A_i(\cdot)$, $\Ecal_i(\cdot)$ and $T_i(\cdot)$
are continuous at $(\gamma_i,\tau_i,\thetabfhat_i,\mu_i)=(\gammabar_{ik},\tau_{ik},\thetabfbar_{ik},\mubar_{ik})$.

\item The denoiser function $g_1(r_1,\gamma_1,\thetabfhat_1)$ and its derivative
 $g_1'(r_1,\gamma_1,\thetabfhat_1)$
are uniformly Lipschitz in $r_1$ at $(\gamma_1,\thetabfhat_1)=(\gammabar_{1k},\thetabfbar_{1k})$.
(See  Appendix~\ref{sec:empConv} for a precise definition of uniform Lipschitz continuity.)

\item The adaptation statistic $\phi_1(r_1,\gamma_1,\thetabfhat_1)$ is
 uniformly pseudo-Lipschitz of order 2 in $r_1$ at
 $(\gamma_1,\thetabfhat_1)=(\gammabar_{1k},\thetabfbar_{1k})$.
\end{enumerate}
Then, for any fixed iteration $k \geq 0$,
\beq \label{eq:aglim}
    \lim_{N \arr \infty} (\alpha_{ik},\eta_{ik},\gamma_{ik},\mu_{ik},\thetabfhat_{ik}) =
    (\alphabar_{ik},\etabar_{ik}, \gammabar_{ik},\mubar_{ik},\thetabfbar_{ik})
\eeq
almost surely.
In addition, the empirical limit \eqref{eq:limrx1} holds almost surely for all $k > 0$,
and \eqref{eq:limqxi} holds almost surely for all $k \geq 0$.
\end{theorem}

Theorem~\ref{thm:em-se} shows that, in the LSL, the parameter estimates $\thetabfhat_{ik}$ converge to
deterministic limits $\thetabfbar_{ik}$ that can be precisely predicted by the
state-evolution equations.  The SE equations incorporate the true distribution of
the components on the prior $\xbf^0$, the true noise precision $\theta_2$,
and the specific parameter estimation and denoiser functions used by the Adaptive VAMP method.
In addition, similar to the SE analysis of VAMP in \cite{rangan2016vamp},
the SE equations also predict the asymptotic joint distribution of $\xbf^0$ and
their estimates $\xbfhat_{ik}$.  This joint distribution can be used to measure
various performance metrics such as MSE~-- see \cite{rangan2016vamp}.
In this way, we have provided a rigorous and precise characterization of a class of adaptive
VAMP algorithms that includes EM-VAMP\@.

\section{Consistent Parameter Estimation with Variance Auto-Tuning} \label{sec:autotune}

By comparing the deterministic limits $\thetabfbar_{ik}$ with the true parameters $\thetabf_i$,
one can determine under which problem conditions
the parameter estimates of adaptive VAMP are asymptotically consistent.
In this section, we show with a particular choice of parameter estimation functions, one
can obtain provably
asymptotically consistent parameter estimates under suitable identifiability conditions.
We call the method \emph{variance auto-tuning}, which
generalizes the approach in \cite{KamRanFU:12-IT}.

\begin{definition} \label{def:ident1}  Let $p(\xbf|\thetabf_1)$ be a parameterized
set of densities.  Given a finite-dimensional statistic
$\phi_1(r)$, consider the mapping
\beq \label{eq:phi1map}
    (\tau_1,\thetabf_1) \mapsto \Exp\left[ \phi_1(R) | \tau_1, \thetabf_1 \right],
\eeq
where the expectation is with respect to the model \eqref{eq:RX0cons}.
We say the $p(\xbf|\thetabf_1)$ is \emph{identifiable} in Gaussian noise
if there exists a finite-dimensional statistic $\phi_1(r) \in \R^d$
such that (i) $\phi_1(r)$ is pseudo-Lipschitz continuous of order 2; and
(ii) the mapping \eqref{eq:phi1map} has a continuous inverse.
\end{definition}

\begin{theorem} \label{thm:theta1cons}
Under the assumptions of Theorem~\ref{thm:em-se},
suppose that $X^0$ follows $X^0 \sim p(\xbf|\thetabf_1^0)$ for some true parameter
$\thetabf_1^0$.  If $p(\xbf|\thetabf_1)$ is identifiable in Gaussian noise,
there exists an adaptation rule
such that, for any iteration $k$, the estimate $\thetabfhat_{1k}$ and noise estimate
$\tauhat_{1k}$ are
asymptotically consistent in that $\lim_{N \arr\infty} \thetabfhat_{1k} = \thetabf_1^0$
and $\lim_{N \arr\infty} \tauhat_{1k} = \tau_{1k}$ almost surely.
\end{theorem}

The theorem is proved in Appendix~\ref{sec:autotuneDetails} which also provides
details on how to perform the adaptation.
The Appendix also provides a similar result for
consistent estimation of the noise precision $\theta_2$.
The result is remarkable as it shows that a simple variant
of EM-VAMP can provide provably consistent parameter estimates under extremely general
distributions.

\section{Numerical Simulations}

\paragraph*{Sparse signal recovery}
We consider a sparse linear inverse problem of
estimating a vector $\xbf$ from measurements $\ybf$ from \eqref{eq:yAx} without knowing the signal parameters $\thetabf_1$ or the noise precision $\theta_2>0$.
The paper \cite{fletcher2016emvamp} presented several
numerical experiments to assess the performance of EM-VAMP relative
to other methods.
To support the results in this paper,
our goal is to demonstrate that state evolution can correctly predict
the performance of EM-VAMP, and to
validate the consistency of EM-VAMP with auto-tuning.
Details are given in Appendix~\ref{sec:sim}.
Briefly, to model the sparsity, $\xbf$ is drawn as an i.i.d.\
Bernoulli-Gaussian (i.e., spike and slab) prior with unknown sparsity level, mean and variance.
The true sparsity is $\beta_x=0.1$.
Following \cite{RanSchFle:14-ISIT,Vila:ICASSP:15}, we take $\Abf \in \R^{M \x N}$ to be a random right-orthogonally invariant matrix with dimensions under $M=512$, $N=1024$
with the condition number set to  $\kappa = 100$ (high condition number
matrices are known to be problem for conventional AMP methods).  The left panel of Fig.~\ref{fig:sim} shows the normalized mean square error (NMSE) for various algorithms.
Appendix~\ref{sec:sim} describes the algorithms in details and also shows similar results for $\kappa=10$.

\begin{figure}
\centering
\includegraphics[width=0.45\columnwidth]{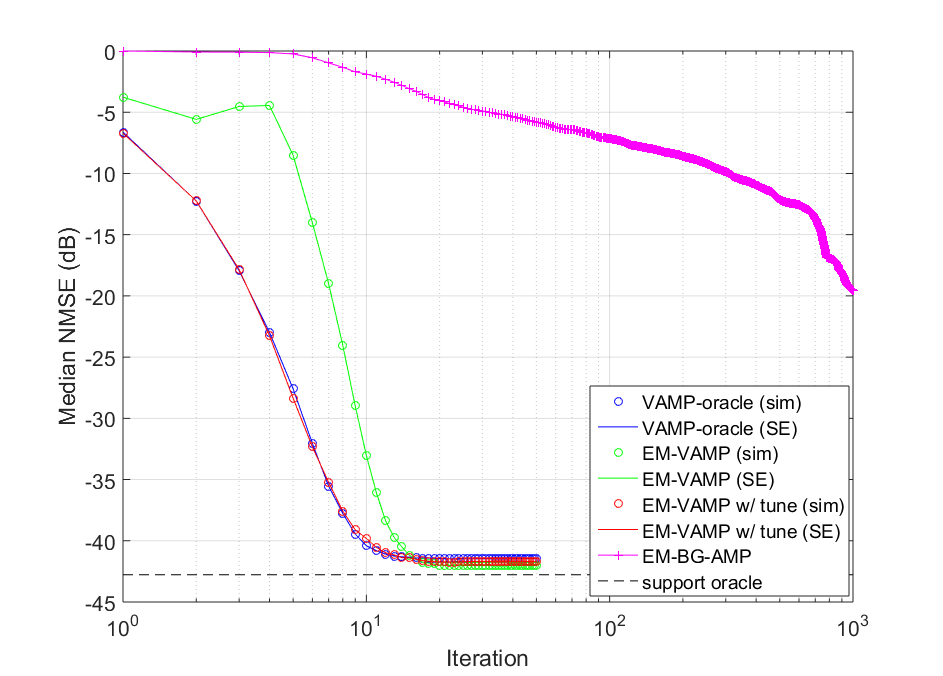}
\hfill
\includegraphics[width=0.45\columnwidth]{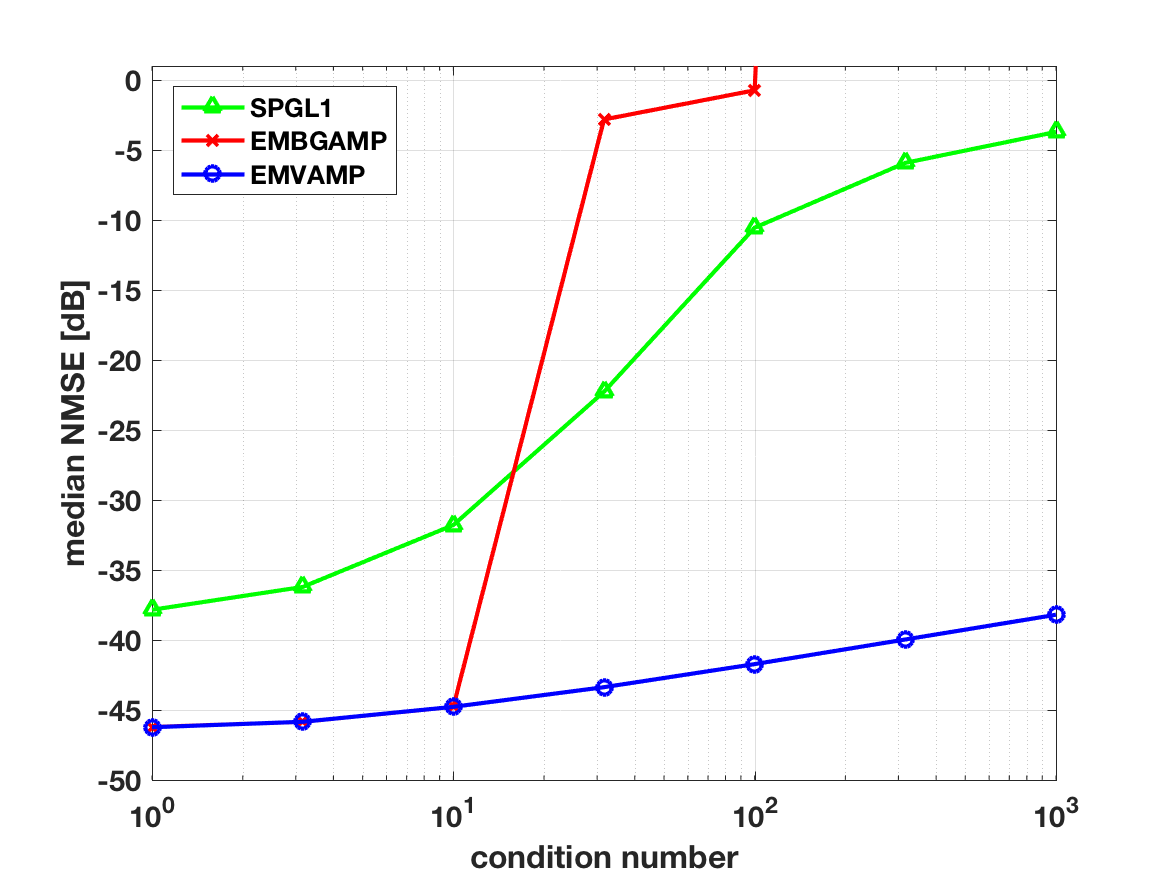}
\caption{Numerical simulations.  Left panel:  Sparse signal recovery:  NMSE versus iteration for condition number for a random matrix
with a condition number $\kappa=100$. Right panel:  NMSE for sparse image recovery (in AWGN at 40~dB SNR) as a function of the condition number $\kappa$.}
\label{fig:sim}
\end{figure}

We see several important features.
First, for all variants of VAMP and EM-VAMP, the SE equations
provide an excellent prediction of the per iteration performance of the algorithm.
Second, consistent with the simulations in \cite{rangan2016vamp},
the oracle VAMP converges remarkably fast ($\sim$ 10 iterations).
Third, the performance of EM-VAMP with
auto-tuning is virtually indistinguishable from oracle VAMP, suggesting that the parameter estimates
are near perfect from the very first iteration.  Fourth,
the EM-VAMP method performs initially worse than the oracle-VAMP, but these errors are exactly
predicted by the SE\@.  Finally, all the VAMP and EM-VAMP algorithms exhibit much faster convergence
than the EM-BG-AMP\@.  In fact, consistent with observations in \cite{fletcher2016emvamp},
EM-BG-AMP begins to diverge at higher condition numbers.  In contrast, the VAMP algorithms are
stable.

\paragraph*{Compressed sensing image recovery}
While the theory is developed on theoretical signal priors, we demonstrate
that the proposed EM-VAMP algorithm can be effective on natural images.  Specifically,
we repeat the experiments in
\cite{vila2014empirical} for recovery of a sparse image.  Again, see Appendix~\ref{sec:sim}
for details including a picture of the image and the various reconstructions.
An $N=256\x 256$ image of a satellite with $K=6678$ non-zero pixels is observed
through a linear transform $\Abf=\diag(\sbf)\Pbf\Hbf\diag(\dbf)$ in AWGN at 40~dB SNR,
where
$\Hbf$ is the fast Hadamard transform,
$\Pbf$ is a random sub-selection to $M<N$ measurements,
$\sbf$ is a scaling to adjust the condition number, and
$\dbf$ is a random $\pm 1$ vector.
As in the previous example,
the image vector $\xbf$ is modeled as sparse Bernoulli-Gaussian and the EM-VAMP
algorithm is used to estimate the sparsity ratio, signal variance, and noise variance.
The transform is set to have measurement ratio $M/N = 0.5$.
We see from the right panel of Fig.~\ref{fig:sim} that the EM-VAMP algorithm
is able to reconstruct the images with improved performance over the standard basis pursuit denoising
method SPGL1 \cite{van2008probing} and the EM-BG-GAMP method from \cite{vila2013expectation}.

\section{Conclusions}  Due to its analytic tractability, computational simplicity,
and potential for Bayes optimal inference, VAMP is a promising technique
for statistical linear inverse problems.
Importantly, it does not suffer the convergence issues of standard AMP methods
when $\Abf$ is ill conditioned.
However, a key challenge in using VAMP and related
methods is the need to precisely specify the
problem parameters.
This work provides a rigorous foundation for analyzing VAMP in combination with various
parameter adaptation techniques including EM\@.  The analysis reveals that VAMP,
with appropriate tuning, can also provide consistent parameter estimates under very
general settings, thus yielding a powerful approach for statistical linear inverse problems.

\appendix

\section{Adaptation with Finite-Dimensional Statistics} \label{sec:finite}

We provide two simple examples where the EM parameter update \eqref{eq:theta1ML}
can be computed from finite-dimensional statistics.

\paragraph*{Finite number of parameter values}
Suppose that $\thetabf_1$ takes only a finite number of values
$\thetabf_1 \in \{ \thetabf_1^1,\ldots,\thetabf_1^L\}$.
Define the statistics,
\[
    \phi_\ell(r_1,\gamma_1,\thetabfhat_1) := \Exp\left[ \ln p(x|\theta_1^\ell)
        | r_1,\gamma_1,\thetahat_1 \right],
\]
where the outer expectation is with respect to the belief estimate in
$b_1(x|r_1,\gamma_1,\thetahat_1)$ given in \eqref{eq:bidef}.
Hence, the EM update in \eqref{eq:theta1ML} is given by
$\thetabfhat_{1,\kp1} = \thetabf_1^{\ell^*}$ where
\[
    \ell^* = \argmax_\ell \mu_{k,\ell}, \quad
    \mu_{\ell} = \frac{1}{N}\sum_{n=1}^N \phi_\ell(r_{1k,n},\gamma_{1k},\thetabfhat_{1k}).
\]

\paragraph*{Exponential family}  Suppose that $p(\xbf|\thetabf_1)$ is given by
an exponential family where the penalty function in \eqref{eq:pxf}
is given by $f_1(x,\thetabf_1) = \thetabf_1\tran\varphi(x)$ for some sufficient
statistics $\varphi(x) \in \R^d$.  Define the
statistic,
\[
    \phi_1(r_1,\gamma_1,\thetabfhat_1) := \Exp\left[ \varphi(x) | r_1,\gamma_1,\thetabfhat_1 \right].
\]
Then, the EM parameter update \eqref{eq:theta1ML} can be computed by
\[
    \thetabfhat_{1,\kp1} = \argmax_{\thetabf_1} \thetabf\tran \mu_1 - \frac{1}{N}\ln Z_1(\thetabf_1),
    \quad
    \mu_1 = \bkt{ \phi_1(\rbf_{1k},\gamma_{1k},\thetabfhat_{1k}) }.
\]
Hence, we see that the EM parameter update can be computed from a finite set of statistics.

\section{Proof of \eqref{eq:theta2T}} \label{sec:theta2TPf}
Using the SVD \eqref{eq:ASVD} and the equations for $\Qbf_k$ and $\xbfhat_{2k}$ in
line~\ref{line:x2} of Algorithm~\ref{algo:em-vamp}, we obtain
\begin{align*}
    \Ubf\tran(\ybf-\Abf\xbfhat_{2k}) = \Ubf\tran \ybf - \Sbf \xbfhat_{2k}
    = \Ubf\tran \diag\left[ \frac{\gamma_{2k}}{\sbf^2\thetahat_{2k} + \gamma_{2k}} \right]
        (\Sbf\qbf_{2k} + \xibf).
\end{align*}
Thus, the parameter estimate $\thetahat_{2,\kp1}$ in line~\ref{line:theta2} is given by
\begin{align*}
    \MoveEqLeft \thetahat_{2,\kp1}^{-1}
        \stackrel{(a)}{=}
        \frac{1}{N}\| \Ubf\tran(\ybf-\Abf\xbfhat_{2k}\|^2 + \frac{1}{N}\Tr(\Abf\tran\Qbf_k\Abf) \\
       &\stackrel{(b)}{=}
        \frac{1}{N}\| \diag\left[\frac{\gamma_{2k}}{\sbf^2\thetahat_{2k} + \gamma_{2k}}\right]
    (\Sbf\qbf_{2k} + \xibf)\|^2 + \frac{1}{N}
        \Tr\left[ \Sbf^2(\Sbf^2\thetahat_{2k} + \gamma_{2k}\Ibf)^{-1} \right],
\end{align*}
where (a) follows since $\Ubf$ is unitary and therefore,
$\|\ybf-\Abf\xbfhat_{2k}\|^2 = \|\Ubf\tran(\ybf-\Abf\xbfhat_{2k})\|^2$ and (b)
follows from the SVD \eqref{eq:ASVD} and some simple manipulations.
Therefore, if we define $\phi_2(\cdot)$ as in \eqref{eq:phi2def}, we obtain \eqref{eq:theta2T}.

\section{Convergence of Vector Sequences} \label{sec:empConv}
We review some definitions from the Bayati and Montanari paper \cite{BayatiM:11}, since we will use the same
analysis framework in this paper.
Fix a dimension $r > 0$, and suppose that, for each $N$,
$\xbf(N)$ is a block vector of the form
\[
    \xbf(N) = (\xbf_1(N),\ldots,\xbf_N(N)),
\]
where each component $\xbf_n(N) \in \R^r$.  Thus, the total dimension
of $\xbf(N)$ is $rN$.  In this case, we will say that
$\xbf(N)$ is a \emph{block vector sequence that scales with $N$
under blocks $\xbf_n(N) \in \R^r$.}
When $r=1$, so that the blocks are scalar, we will simply say that
$\xbf(N)$ is a \emph{vector sequence that scales with $N$}.
Such vector sequences can be deterministic or random.
In most cases, we will omit the notational dependence on $N$ and simply write $\xbf$.

Now, given $p \geq 1$,
a function $\fbf:\R^s \arr \R^r$ is called \emph{pseudo-Lipschitz of order $p$},
if there exists a constant $C > 0$ such that for all $\xbf_1,\xbf_2 \in\R^s$,
\[
    \|\fbf(\xbf_1)-\fbf(\xbf_2)\| \leq C\|\xbf_1-\xbf_2\|\left[ 1 + \|\xbf_1\|^{p-1}
    + \|\xbf_2\|^{p-1} \right].
\]
Observe that in the case $p=1$, pseudo-Lipschitz continuity reduces to
the standard Lipschitz continuity.

Given $p \geq 1$, we will say that the block vector sequence $\xbf=\xbf(N)$ converges
\emph{empirically with $p$-th order moments} if there exists a random variable
$X \in \R^r$ such that
\begin{enumerate}[(i)]
\item $\Exp|X|^p < \infty$; and
\item for any scalar-valued pseudo-Lipschitz continuous function $f$ of order $p$,
\[
    \lim_{N \arr \infty} \frac{1}{N} \sum_{n=1}^N \fbf(\xbf_n(N)) = \Exp\left[ \fbf(X) \right].
\]
\end{enumerate}
Thus, the empirical mean of the components $\fbf(\xbf_n(N))$ converges to
the expectation $\Exp( \fbf(X) )$.
In this case, with some abuse of notation, we will write
\beq \label{eq:plLim}
    \lim_{N \arr \infty} \left\{ \xbf_n \right\} \stackrel{PL(p)}{=} X,
\eeq
where, as usual, we have omitted the dependence $\xbf_n=\xbf_n(N)$.
Importantly, empirical convergence can be defined on deterministic vector sequences,
with no need for a probability space.  If $\xbf=\xbf(N)$ is a random vector sequence,
we will often require that the limit \eqref{eq:plLim} holds almost surely.

Finally, let $\fbf(\xbf,\lambda)$ be a function on $\xbf \in \R^s$, dependent on some parameter
vector $\lambda \in \R^d$.
We say that $\fbf(\xbf,\lambda)$ is \emph{uniformly pseudo-Lipschitz continuous} of order $p$
in $\xbf$ at $\lambda=\lambdabar$ if there exists an open neighborhood $U$ of $\lambdabar$, such that
(i) $\fbf(\xbf,\lambda)$ is pseudo-Lipschitz of order $p$ in $\xbf$
for all $\lambda \in U$; and (ii)
there exists a constant $L$ such that
\beq \label{eq:unifLip2}
    \|\fbf(\xbf,\lambda_1)-\fbf(\xbf,\lambda_2)\| \leq L\left(1+\|\xbf\|\right)|\lambda_1-\lambda_2|,
\eeq
for all $\xbf \in \R^s$ and $\lambda_1,\lambda_2 \in U$.
It can be verified that if $\fbf(\xbf,\lambda)$
is uniformly pseudo-Lipschitz continuous of order $p$
in $\xbf$ at $\lambda=\lambdabar$, and $\lim_{N \arr \infty} \lambda(N) = \lambdabar$, then
\[
    \lim_{N \arr \infty} \xbf \stackrel{PL(p)}{=} X \Longrightarrow
    \lim_{N \arr \infty} \lim_{N \arr \infty} \frac{1}{N} \sum_{n=1}^N \fbf(\xbf_n(N),\lambda(N))
    = \Exp\left[ \fbf(X,\lambdabar) \right].
\]
In the case when $\fbf(\cdot)$ is uniformly pseudo-Lipschitz continuous of order $p=1$,
we will simply say $\fbf(\cdot)$ is uniformly Lipschitz.

\section{A General Convergence Result}

The proof of Theorem~\ref{thm:em-se} uses the following minor modification of the
general convergence result in \cite[Theorem 4]{rangan2016vamp}:
We are given a dimension $N$,
an orthogonal matrix $\Vbf \in \R^{N \x N}$,
an initial vector $\ubf_0 \in \R^N$, and disturbance vectors
$\wbf^p, \wbf^q \in \R^N$.
Then, we generate a sequence of iterates by the following recursion:%
\begin{subequations}
\label{eq:algoGen}
\begin{align}
    \pbf_k &= \Vbf\ubf_k \label{eq:pupgen} \\
    \alpha_{1k} &= \bkt{ \fbf_p'(\pbf_k,\wbf^p,\gamma_{1k},\thetabfhat_{1k})},
    \quad \mu_{1k} = \bkt{ \phibf_p(\pbf_k,\wbf^p,\gamma_{1k},\thetabfhat_{1k}) } \label{eq:alpha1gen}  \\
    (C_{1k},\gamma_{2k},\thetabfhat_{1,\kp1}) &= \Gamma_1(\gamma_{1k},\alpha_{1k},\mu_{1k},\thetabfhat_{1k})
        \label{eq:gam1gen} \\
    \vbf_k &= C_{1k}\left[
        \fbf_p(\pbf_k,\wbf^p,\gamma_{1k},\thetabfhat_{1k})- \alpha_{1k} \pbf_{k} \right] \label{eq:vupgen} \\
    \qbf_k &= \Vbf\tran\vbf_k \label{eq:qupgen} \\
    \alpha_{2k} &= \bkt{ \fbf_q'(\qbf_k,\wbf^q,\gamma_{2k},\thetahat_{2k}))}, \quad
    \mu_{2k} = \bkt{ \phi_q(\qbf_k,\wbf^q,\gamma_{2k},\thetahat_{2k}) }         \label{eq:alpha2gen} \\
    (C_{2k},\gamma_{1,\kp1},\thetahat_{2,\kp1}) &= \Gamma_1(\gamma_{2k},\alpha_{2k},\mu_{2k},\thetahat_{2k})
        \label{eq:gam2gen} \\
    \ubf_{\kp1} &= C_{2k}\left[
        \fbf_q(\qbf_k,\wbf^q,\gamma_{2k},\thetahat_{2k}) - \alpha_{2k}\qbf_{k} \right], \label{eq:uupgen}
\end{align}
\end{subequations}
which is initialized with some vector $\ubf_0$, scalar $\gamma_{10}$ and parameters
$\thetahat_0 = (\thetabfhat_{10},\thetahat_{20})$.
Here, $\fbf_p(\cdot)$, $\fbf_q(\cdot)$, $\phibf_p(\cdot)$ and $\phibf_q(\cdot)$
are separable functions, meaning for all components $n$,
\begin{align}\label{eq:fpqcomp}
\begin{split}
    \left[ \fbf_p(\pbf,\wbf^p,\gamma_1,\thetabfhat_1)\right]_n = f_p(p_n,w^p_n,\gamma_1,\thetabfhat_1) \quad
    \left[ \fbf_q(\qbf,\wbf^q,\gamma_2,\thetahat_2)\right]_n = f_q(q_n,w^q_n,\gamma_2,\thetahat_2), \\
    \left[ \phibf_p(\pbf,\wbf^p,\gamma_1,\thetabfhat_1)\right]_n = \phi_p(p_n,w^p_n,\gamma_1,\thetabfhat_1) \quad
    \left[ \phibf_q(\qbf,\wbf^q,\gamma_2,\thetahat_2)\right]_n = \phi_q(q_n,w^q_n,\gamma_2,\thetahat_2), \\
\end{split}
\end{align}
for functions $f_p(\cdot)$, $f_q(\cdot)$, $\phi_q(\cdot)$ and $\phi_p(\cdot)$. These functions may be
vector-valued, but their dimensions are fixed and do not change with $N$.
The functions $\Gamma_i(\cdot)$, $C_i(\cdot)$ and $T_i(\cdot)$
also have fixed dimensions with $C_i(\cdot)$ being scalar-valued.

Similar to the analysis of EM-VAMP, we consider the following
large-system limit (LSL) analysis.
We consider a sequence of runs of the recursions indexed by $N$.
We model the initial condition $\ubf_0$ and disturbance vectors $\wbf^p$ and $\wbf^q$
as deterministic sequences that scale with $N$ and assume that their components
converge empirically as
\beq \label{eq:U0lim}
    \lim_{N \arr \infty} \{ u_{0n} \} \PLeq U_0, \quad
    \lim_{N \arr \infty} \{ w^p_n \} \PLeq W^p, \quad
    \lim_{N \arr \infty} \{ w^q_n \} \PLeq W^q,
\eeq
to random variables $U_0$, $W^p$ and $W^q$.  We assume that the
initial constants converges as
\beq \label{eq:gam10limgen}
    \lim_{N \arr \infty} (\gamma_{10},\thetabfhat_{10},\thetahat_{20}) = (\gammabar_{10},\thetabfbar_{10},\thetabar_{20})
\eeq
for some values $(\gammabar_{10},\thetabfbar_{10},\thetabar_{20})$.
The matrix $\Vbf \in \R^{N \x N}$
is assumed to be Haar distributed (i.e.\ uniform on the set of orthogonal matrices) and
independent of $\rbf_0$, $\wbf^p$ and $\wbf^q$.
Since $\rbf_0$, $\wbf^p$ and $\wbf^q$ are deterministic, the only randomness is in the matrix $\Vbf$.

Under the above assumptions, define the SE equations
\begin{subequations} \label{eq:segen}
\begin{align}
    \alphabar_{1k} &= \Exp\left[ f_p'(P_k,W^p,\gammabar_{1k},\thetabfbar_{1k})\right],
    \quad
    \mubar_{1k} = \Exp\left[ \phi_p(P_k,W^p,\gammabar_{1k},\thetabfbar_{1k}) \right] \label{eq:a1segen}  \\
    (\Cbar_{1k},\gammabar_{2k},\thetabfbar_{1,\kp1}) &=
        \Gamma_1(\gammabar_{1k},\alphabar_{1k},\mubar_{1k},\thetabfbar_{1k})
        \label{eq:gam1segen} \\
    \tau_{2k} &= \Cbar_{1k}^2 \left\{
        \Exp\left[ f_p^2(P_k,W^p,\gammabar_{1k},\thetabfbar_{1k})\right]  - \alphabar_{1k}^2\tau_{1k} \right\}
        \label{eq:tau2segen} \\
    \alphabar_{2k} &= \Exp\left[ f_q'(Q_k,W^q,\gammabar_{2k},\thetabar_{2k}) \right],
    \quad
    \mubar_{1k} = \Exp\left[ \phi_q(Q_k,W^q,\gammabar_{2k},\thetabar_{2k}) \right]
        \label{eq:a2segen} \\
    \tau_{1,\kp1} &= \Cbar_{2k}^2\left\{
        \Exp\left[ f_q^2(Q_k,W^q,\gammabar_{2k},\thetabar_{2k})\right] - \alphabar_{2k}^2\tau_{2k}\right\}
        \label{eq:tau1segen}
\end{align}
\end{subequations}
which are initialized with the values in \eqref{eq:gam10limgen} and
\beq \label{eq:tau10gen}
    \tau_{10} = \Exp U_0^2,
\eeq
where $U_0$ is the random variable in \eqref{eq:U0lim}.
In the SE equations~\eqref{eq:segen},
the expectations are taken with respect to random variables
\[
    P_k \sim \Norm(0,\tau_{1k}), \quad Q_k \sim \Norm(0,\tau_{2k}),
\]
where $P_k$ is independent of $W^p$ and $Q_k$ is independent of $W^q$.

\begin{theorem} \label{thm:genConv}  Consider the recursions \eqref{eq:algoGen}
and SE equations \eqref{eq:segen} under the above assumptions.  Assume additionally that,
for all $k$:
\begin{enumerate}[(i)]
\item For $i=1,2$, the function $\Gamma_i(\gamma_i,\alpha_i,\mu_i,\thetahat_i)$
is continuous at the points $(\gammabar_{ik},\alphabar_{ik},\mubar_{ik},\thetabfbar_{ik})$
from the SE equations;
\item The function $f_p(p,w^p,\gamma_1,\thetabfhat_1)$ and its derivative $f_p'(p,w^p,\gamma_1)$
are uniformly Lipschitz continuous in $(p,w^p)$ at $(\gamma_1,\thetabf_1)=(\gammabar_{1k},\thetabfbar_{1k})$.
The function $f_q(q,w^q,\gamma_2,\thetahat_2)$ satisfies the analogous conditions.

\item The function $\phi_p(p,w^p,\gamma_1,\thetabfhat_1)$ is
uniformly pseudo Lipschitz continuous of order 2
in $(p,w^p)$ at $(\gamma_1,\thetabf_1)=(\gammabar_{1k},\thetabfbar_{1k})$.
The function $\phi_q(q,w^q,\gamma_2,\thetahat_2)$ satisfies the analogous conditions.
\end{enumerate}
Then,
\begin{enumerate}[(a)]
\item For any fixed $k$, almost surely the components of $(\wbf^p,\pbf_0,\ldots,\pbf_k)$
empirically converge as
\beq \label{eq:Pconk}
    \lim_{N \arr \infty} \left\{ (w^p_n,p_{0n},\ldots,p_{kn}) \right\}
    \PLeq (W^p,P_0,\ldots,P_k),
\eeq
where $W^p$ is the random variable in the limit \eqref{eq:U0lim} and
$(P_0,\ldots,P_k)$ is a zero mean Gaussian random vector independent of $W^p$,
with $\Exp(P_k^2) = \tau_{1k}$.  In addition, we have that
\beq \label{eq:ag1limgen}
    \lim_{N \arr \infty} (\alpha_{1k},\mu_{1k},\gamma_{2k},C_{1k},\thetabfhat_{1,\kp1})
        = (\alphabar_{1k},\mubar_{1k},\gammabar_{2k},\Cbar_{1k},\thetabfbar_{1,\kp1})
\eeq
almost surely.

\item For any fixed $k$, almost surely the components of $(\wbf^q,\qbf_0,\ldots,\qbf_k)$
empirically converge as
\beq \label{eq:Qconk}
    \lim_{N \arr \infty} \left\{ (w^q_n,q_{0n},\ldots,q_{kn}) \right\}
    \PLeq (W^q,Q_0,\ldots,Q_k),
\eeq
where $W^q$ is the random variable in the limit \eqref{eq:U0lim} and
$(Q_0,\ldots,Q_k)$ is a zero mean Gaussian random vector independent of $W^q$,
with $\Exp(P_k^2) = \tau_{2k}$.  In addition, we have that
\beq \label{eq:ag2limgen}
    \lim_{N \arr \infty} (\alpha_{2k},\mu_{2k},\gamma_{1,\kp1},C_{2k},\thetahat_{2,\kp1})
        = (\alphabar_{2k},\mubar_{2k},\gammabar_{1,\kp1},\Cbar_{2k},\thetabar_{2,\kp1})
\eeq
almost surely.
\end{enumerate}
\end{theorem}
\begin{proof}  This is a minor modification of ~\cite[Theorem 4]{rangan2016vamp}.
Indeed, the only change here is the addition of the terms $\mu_{ik}$ and $\thetabfhat_{ik}$.
Their convergence can be proven identically to that of $\alpha_{ik}$, $\gamma_{ik}$.
For example, consider the proof of convergence $\alpha_{1k} \arr \alphabar_{1k}$
in \cite[Theorem 4]{rangan2016vamp}.  This is proved as part of an induction argument
where, in that part of the proof, the induction hypothesis is that \eqref{eq:Pconk} holds for some $k$
and \eqref{eq:ag2limgen} holds for $\km1$.  Thus, $\gamma_{1k} \arr \gammabar_{1k}$
and $\thetabfhat_{1k} \arr \thetabfbar_{1k}$ almost surely and $\lim_{N \arr \infty} (\pbf_k,\wbf^p) \PLeq
(P_k,W^p)$.  Hence, since
$f_p(\cdot)$ is uniformly Lipschitz continuous,
\[
    \lim_{N \arr \infty} \alpha_{1k} = \lim_{N \arr \infty}
        \bkt{ \fbf_p(\pbf_k,\wbf^p,\gamma_{1k},\thetabfhat_{1k}) }
        = \Exp\left[ f_p(P_k,W^p,\gammabar_{1k},\thetabfbar_{1k}) \right] = \alphabar_{1k}.
\]
Similarly, since we have assumed the uniform pseudo-Lipschitz continuity of $\phi_p(\cdot)$,
\[
    \lim_{N \arr \infty} \mu_{1k} = \lim_{N \arr \infty}
        \bkt{ \phibf_p(\pbf_k,\wbf^p,\gamma_{1k},\thetabfhat_{1k}) }
        = \Exp\left[ \phi_p(P_k,W^p,\gammabar_{1k},\thetabfbar_{1k}) \right] = \mubar_{1k}.
\]
Thus, we have
\[
    \lim_{N \arr \infty} (\gamma_{1k},\alpha_{1k},\mu_{1k},\thetabfhat_{1k})
    = (\gammabar_{1k},\alphabar_{1k},\mubar_{1k},\thetabfbar_{1k}),
\]
almost surely.  Then, since $\Gamma_1(\cdot)$ is continuous, \eqref{eq:gam1gen} shows that
\[
    \lim_{N \arr \infty} (C_{1k},\gamma_{2k},\thetabfhat_{1,\kp1}) =
    (\Cbar_{1k},\gammabar_{2k},\thetabfbar_{1,\kp1}) =
        \Gamma_1(\gammabar_{1k},\alphabar_{1k},\mubar_{1k},\thetabfbar_{1k}).
\]
This proves \eqref{eq:ag1limgen}.  The modifications for the proof of
\eqref{eq:ag2limgen} are similar.
\end{proof}

\section{Proof of Theorem~\ref{thm:em-se} }

This proof is also a minor modification of the SE analysis of VAMP in
\cite[Theorem 1]{rangan2016vamp}.
As in that proof,
we need to simply rewrite the recursions in Algorithm~\ref{algo:em-vamp} in the form
\eqref{eq:algoGen} and apply the general convergence result in Theorem~\ref{thm:genConv}.
To this end, define the error terms
\beq \label{eq:pvslr}
    \pbf_k := \rbf_{1k}-\xbf^0, \quad
    \vbf_k := \rbf_{2k}-\xbf^0,
\eeq
and their transforms,
\beq \label{eq:uqslr}
    \ubf_k := \Vbf\tran\pbf_k, \quad
    \qbf_k := \Vbf\tran\vbf_k.
\eeq
Also, define the disturbance terms
\beq \label{eq:wpqslr}
    \wbf^q := (\xibf,\sbf), \quad
    \wbf^p := \xbf^0, \quad \xibf := \Ubf\tran\wbf.
\eeq
Define the component-wise update functions
\beq \label{eq:fqpslr}
    f_p(p,x^0,\gamma_1,\thetabfhat_1) := g_1(p+x^0,\gamma_1,\thetabfhat_1) - x^0, \quad
    f_q(q,(\xi,s),\gamma_2,\thetahat_2) := \frac{\thetahat_2 s\xi + \gamma_2 q}{
        \thetahat_2 s^2 + \gamma_2}.
\eeq
Also, define the adaptation functions,
\beq \label{eq:phiqpslr}
    \phi_p(p,x^0,\gamma_1,\thetabfhat_1) := \phi_1(p+x^0,\gamma_1,\thetabfhat_1), \quad
    \phi_q(q,(\xi,s),\gamma_2,\thetahat_2) := \phi_2(q,\xi,s,\gamma_2,\thetahat_2),
\eeq
where $\phi_1(\cdot)$ is the statistic function in Algorithm~\ref{algo:em-vamp}
and $\phi_2(\cdot)$ are defined in \eqref{eq:phi2def}.

With these definitions, we claim that the outputs satisfy the recursions:
\begin{subequations} \label{eq:gecslr}
\begin{align}
    \pbf_k &= \Vbf\ubf_k \label{eq:pupslr} \\
    \alpha_{1k} &= \bkt{ \fbf_p'(\pbf_k,\xbf^0,\gamma_{1k},\thetabfhat_1)},
    \quad \gamma_{2k} = \frac{(1-\alpha_{1k})\gamma_{1k}}{\alpha_{1k}}
        \label{eq:alpha1slr}  \\
    \vbf_k &= \frac{1}{1-\alpha_{1k}}\left[
        \fbf_p(\pbf_k,\xbf^0,\gamma_{1k})- \alpha_{1k} \pbf_{k} \right] \label{eq:vupslr} \\
    \thetabfhat_{1,\kp1} &= T_1(\mu_{1k}), \quad \mu_{1k} = \bkt{ \phi_1(\vbf_{k} + \xbf^0,\gamma_{1k},\thetabfhat_{1k}) }
        \label{eq:theta1slr} \\
    \qbf_k &= \Vbf\tran\vbf_k \label{eq:qupslr} \\
    \alpha_{2k} &= \bkt{ \fbf_q'(\qbf_k,\wbf^q,\gamma_{2k})},
    \quad \gamma_{1,\kp1} = \frac{(1-\alpha_{2k})\gamma_{2k}}{\alpha_{2k}}
        \label{eq:alpha2slr} \\
    \ubf_{\kp1} &= \frac{1}{1-\alpha_{2k}}\left[
        \fbf_q(\qbf_k,\wbf^q,\gamma_{2k}) - \alpha_{2k}\qbf_{k} \right] \label{eq:uupslr} \\
    \thetahat_{2,\kp1} &= T_2(\mu_{2k}), \quad \mu_{2k} = \bkt{ \phi_2(\qbf_{k},\xibf,\gamma_{2k},\thetahat_{2k}) }
        \label{eq:theta2slr}
\end{align}
\end{subequations}
The fact that the updates in the modified Algorithm~\ref{algo:em-vamp}
satisfy \eqref{eq:pupslr}, \eqref{eq:alpha1slr}, \eqref{eq:qupslr},
\eqref{eq:alpha2slr} can be proven exactly as in the proof of \cite[Theorem 1]{rangan2016vamp}.
The update \eqref{eq:theta1slr} follows from line~\ref{line:theta1} in Algorithm~\ref{algo:em-vamp}.
The update \eqref{eq:theta2slr} follows from the definitions in \eqref{eq:theta2T}.

Now, define
\[
    C_{ik} = \frac{1}{1-\alpha_i},
\]
and the function,
\beq \label{eq:Gamdefpf}
    \Gamma_i(\gamma_{i},\alpha_{i},\mu_{i},\thetabfhat_{i}) = \left(
        \frac{1}{1-\alpha_i}, \frac{\gamma_i(1-\alpha_i)}{\alpha_i},
        T_i(\mu_i) \right),
\eeq
the updates \eqref{eq:gecslr} are equivalent to
the general recursions in \eqref{eq:algoGen}.
Also, the continuity assumptions (i)--(iii)
in Theorem~\ref{thm:genConv} are also easily
verified.  For example, for assumption (i), $\Gamma_i(\cdot)$ in \eqref{eq:Gamdefpf}
is continuous, since we have assumed that $\alphabar_{ik} \in (0,1)$ and
$T_i(\mu_i)$ is continuous at $\mu_i = \mubar_{ik}$.

For assumption (ii), the function $f_p(\cdot)$ in \eqref{eq:fqpslr} is uniformly
Lipschitz-continuous since we have assumed that $g_1(\cdot)$ is uniformly Lipschitz-continuous.
Also, the SE equations \eqref{eq:eta1se} and \eqref{eq:eta1se} and the assumption that
$\alphabar_{ik} > 0$ show that $\gammabar_{ik} > 0$ for all $i$ and $k$.
In addition, the definition of $T_2(\cdot)$ and $\phi_2(\cdot)$ in \eqref{eq:theta2T}
and \eqref{eq:phi2def} shows that $\thetabar_{2k} > 0$ as defined in \eqref{eq:theta2se}.
Therefore, since $s$ is bounded in $[0,s_{\rm max}]$ and $\gammabar_{2k} > 0$, $\thetabar_{2k} > 0$,
we see that $f_q(\cdot)$ in \eqref{eq:fqpslr} is uniformly Lipschitz-continuous.

Finally, for assumption (iii) in Theorem~\ref{thm:genConv}, the uniform pseudo-Lipschitz continuity
of $\phi_p(\cdot)$ in \eqref{eq:phiqpslr} follows from the
uniform pseudo-Lipschitz continuity assumption on $\phi_1(\cdot)$.  Also, $\phi_2(\cdot)$
in \eqref{eq:phi2def} is
uniformly pseudo-Lipschitz continuous of order 2 since $s \in [0,s_{\rm max}]$.
Therefore,  $\phi_q(\cdot)$ in \eqref{eq:phiqpslr} is also
uniformly pseudo-Lipschitz continuous of order 2.

Hence, all the conditions of Theorem~\ref{thm:genConv} are satisfied.  The SE equations
\eqref{eq:se} then follow directly from the SE equations \eqref{eq:segen}.
This completes the proof.

\section{Variance Auto-Tuning Details and Proof of Theorem~\ref{thm:theta1cons}}
\label{sec:autotuneDetails}

\subsection{Proof of Theorem~\ref{thm:theta1cons}}
To prove Theorem~\ref{thm:theta1cons}, we perform the \emph{variance auto-tuning}
adaptation as follows.
Suppose the $X^0$ follows $X^0 \sim p(\xbf|\thetabf_1^0)$ for some true parameter
$\thetabf_1^0$.  Then, the SE analysis shows that the components
of $\rbf_{1k}$ asymptotically behave as
\beq \label{eq:RX0cons}
    R = X^0 + P, \quad P = \Norm(0,\tau_1), \quad X^0 \sim p(\xbf|\thetabf_1),
\eeq
for $\thetabf_1 = \thetabf_1^0$ and some $\tau_1 = \tau_{1k}$.
The idea in variance auto-tuning is to jointly estimate $(\thetabf_1,\tau_1)$, say via maximum
likelihood estimation,
\beq \label{eq:theta1ML}
    (\tauhat_{1k},\thetabfhat_{1k}) = \argmax_{\tau_1,\thetabf_1}
        \sum_{n=1}^N \ln p(r_{1k,n}| \tau_1,\thetabf_1),
\eeq
where we use the notation $p(r_1|\tau_1,\thetabf_1)$ to denote the density of $R$
under the model \eqref{eq:RX0cons}.
We can then use both the parameter estimate $\thetabfhat_{1k}$ and variance
estimate $\tauhat_{1k}$ in the denoiser,
\beq \label{eq:xhat1auto}
    \xbfhat_{1k} = \gbf_1(\rbf_{1k},\gamma_{1k},\thetabfhat_{1k}), \quad
    \gamma_{1k} = \tauhat_{1k}^{-1}.
\eeq
The parameter estimate \eqref{eq:theta1ML} and precision estimate $\gamma_{1k} = \tauhat_{1k}^{-1}$
would replace lines~\ref{line:theta1} and \ref{line:gam1} in Algorithm~\ref{algo:em-vamp}.
Since we are using the estimated precision $\tauhat_{1k}^{-1}$ from \eqref{eq:theta1ML}
instead of the
predicted precision $\gamma_{1k}$, we call the estimation \eqref{eq:theta1ML}
with the estimate \eqref{eq:xhat1auto} variance auto-tuning.

Now suppose that statistic as in Definition~\ref{def:ident1} exists.
By the inverse property,
there must be a continuous mapping $T_1(\mu)$ such that for any $\tau_1,\thetabf_1$,
\beq \label{eq:muT1}
    \mu_1 = \Exp\left[  \phi_1(R) | \tau_1, \thetabf_1 \right] \Rightarrow
    T_1(\mu_1) = (\tau_1,\thetabf_1).
\eeq
Although the order of updates of $\thetabfhat_{1k}$ and $\gamma_{1k}$ is slightly
changed from Algorithm~\ref{algo:em-vamp}, we can apply an identical SE analysis
and show that, in the LSL,
\beq \label{eq:muT2}
    \lim_{N \arr \infty} (\tauhat_{1k}, \thetabfhat_{1k}) = (\taubar_{1k},\thetabfhat_{1k})
    = T_1(\mubar_{1k}),
    \quad
    \mubar_{1k} = \Exp\left[  \phi_1(R_1) | \tau_{1k}, \thetabf_1^0 \right],
\eeq
where $\thetabf_1^0$ is the true parameter.  Comparing
\eqref{eq:muT1} and \eqref{eq:muT2} we see that $\thetabfbar_{1k} = \thetabf_1^0$
and $\taubar_{1k} = \tau_{1k}$.
This proves Theorem~\ref{thm:theta1cons}.

\subsection{Auto-Tuning Details}
The auto-tuning above requires solving a
joint ML estimation \eqref{eq:theta1ML}.
We describe a simple approximate method for this estimation using EM\@.
First, we rewrite the ML estimation in \eqref{eq:theta1ML} as
\beq \label{eq:theta1MLvec}
    (\tauhat_1,\thetabfhat_1) = \argmax_{\tau_1,\thetabf_1} \ln p(\rbf_1|\tau_1,\thetabf_1),
\eeq
where we have dropped the dependence on the iteration $k$ to simplify the notation and we have used
the notation $p(\rbf_1|\tau_1,\thetabf_1) := \prod_n p(r_{1n}|\tau_1,\thetabf_1)$.
It will be more convenient to precisions rather than variances, so we rewrite
\eqref{eq:theta1MLvec2},
\beq \label{eq:theta1MLvec2}
    (\gammahat_1,\thetabfhat_1) = \argmax_{\gamma_1,\thetabf_1} \ln p(\rbf_1|\gamma_1,\thetabf_1),
\eeq
where $\gamma_1 = 1/\tau_1$.
Treating $\xbf$ as latent vector, the EM procedure for the ML estimate
\eqref{eq:theta1MLvec2} performs updates $(\gammahat_1^{\rm old},\thetabfhat_1^{\rm old})
\mapsto (\gammahat_1^{\rm new},\thetabfhat_1^{\rm new})$ given by
\beq \label{eq:EMauto}
    (\gammahat_1^{\rm new},\thetabfhat_1^{\rm new})
        = \argmax_{\gamma_1,\thetabf_1} \Exp\left[ \ln p(\xbf,\rbf_1|\gamma_1,\thetabf_1)
            \mid \rbf_1, \gamma_1^{\rm old},\thetabfhat_1^{\rm old} \right],
\eeq
where the expectation is with respect to the posterior density,
$p(\xbf|\rbf_1,\gammahat_1^{\rm old},\thetabfhat_1^{\rm old})$.  But, this posterior density
is precisely the belief estimate $b_i(\xbf|\rbf_1,\gamma_i^{\rm old},\thetahat_1^{\rm old})$.
Now, under the model~\eqref{eq:RX0cons},
the log likelihood $\ln p(\xbf,\rbf_1|\gamma,\thetabf_1)$ is given by
\begin{align}
    \MoveEqLeft \ln p(\xbf,\rbf_1|\gamma_1,\thetabf_1)
        = \ln p(\rbf|\xbf,\gamma_1) + \ln p(\xbf|\thetabf_1) \nonumber \\
        &= -\frac{\gamma_1}{2}\|\rbf-\xbf\|^2 + \frac{N}{2}\ln(\gamma_1) + \ln p(\xbf|\thetabf_1).
        \nonumber
\end{align}
Thus, the maximization in \eqref{eq:EMauto} with respect to $\gamma_1$ is
\begin{align}
    \left(\gammahat_1^{\rm new} \right)^{-1} = \frac{1}{N}
        \Exp\left[ \|\xbf-\rbf_1\|^2 \mid| \tauhat_1^{\rm old},\thetabfhat_1^{\rm old} \right]
        = \frac{1}{N} \|\xbfhat_1 -\rbf_1\|^2 + \eta_1^{-1},
\end{align}
where
\[
    \xbfhat_1 = \Exp\left[ \xbf \mid \rbf_1, \gammahat_1^{\rm old},\thetabfhat_1^{\rm old} \right],
    \quad
    \eta_1^{-1} = \frac{1}{N}
        \Tr\Cov\left[ \xbf \mid \rbf_1, \gammahat_1^{\rm old},\thetabfhat_1^{\rm old} \right].
\]
The maximization in \eqref{eq:EMauto} with respect to $\thetabf_1$ is
\[
    \thetabfhat_1^{\rm new} = \argmax_{\thetabf_1} \Exp\left[ p(\xbf|\thetabf_1) |
        \rbf_1,\gamma_1^{\rm old},\thetabfhat_1^{\rm old} \right],
\]
which is precisely the EM update \eqref{eq:thetaEM} with the current parameter estimates.
Thus, the auto-tuning estimation can be performed by repeatedly performing updates
\begin{subequations}
\begin{align}
    \xbfhat_1 &= \Exp\left[ \xbf \mid \rbf_1, \gammahat_1^{\rm old},\thetabfhat_1^{\rm old} \right],
    \quad
    \eta_1^{-1} = \frac{1}{N}
        \Tr\Cov\left[ \xbf \mid \rbf_1, \gammahat_1^{\rm old},\thetabfhat_1^{\rm old} \right],
        \label{eq:xhatauto} \\
    \left(\gammahat_1^{\rm new} \right)^{-1} &= \frac{1}{N} \|\xbfhat_1 -\rbf_1\|^2 + \eta_1^{-1}
        \label{eq:gamauto} \\
    \thetabfhat_1^{\rm new} &= \argmax_{\thetabf_1} \Exp\left[ p(\xbf|\thetabf_1) |
        \rbf_1,\gamma_1^{\rm old},\thetabfhat_1^{\rm old} \right].
        \label{eq:thetaauto}
\end{align}
\end{subequations}
Thus, each iteration of the auto-tuning involves calling the denoiser \eqref{eq:xhatauto} at
the current estimate; updating the variance from the denoiser error \eqref{eq:gamauto}; and
updating the parameter estimate via an EM-type update \eqref{eq:thetaauto}.

\subsection{Consistent Estimation of $\theta_2$} \label{sec:theta2Est}

We show how we can perform a similar ML estimation as in Section~\ref{sec:autotune}
for estimating the noise precision $\theta_2$ and error variance $\tau_2$ at the output.
Suppose that the true noise precision is $\theta_2 = \theta_2^0$.
Define
\beq \label{eq:zkpf}
    \zbf_k = \Sbf\Vbf\tran\rbf_{2k} - \Ubf\tran\ybf .
\eeq
Then,
\[
    \zbf_k \stackrel{(a)}{=} \Sbf\Vbf\tran\rbf_{2k} - \Ubf\tran(\Abf\xbf^0 + \wbf)
    \stackrel{(b)}{=} \Sbf\qbf_k + \xibf,
\]
where (a) follows from \eqref{eq:yAxslr} and (b) follows from \eqref{eq:ASVD} and \eqref{eq:qerrdef}.
From Theorem~\ref{thm:em-se} and the limit \eqref{eq:limqxi}, $\zbf_k$ converge empirically to
a random variable $Z_k$ given by
\beq \label{eq:zmod}
    Z_k = SQ_k + \Xi, \quad Q_k \sim \Norm(0,\tau_2), \quad \Xi \sim \Norm(0,\theta_2^{-1}),
\eeq
where $(\theta_2,\tau_2)=(\theta_2^0,\tau_{2k})$ are the true noise precision and error variance.
Similar to the auto-tuning of $\thetabf_1$ in Section~\ref{sec:autotune}, we
can attempt to estimate $(\tau_2,\theta_2)$ via ML estimation,
\beq \label{eq:theta2ML}
    (\tauhat_{2k},\thetahat_{2k}) = \argmax_{\tau_2,\theta_2}
    \frac{1}{N} \sum_{n=1}^N \ln p(z_{k,n}|\theta_2,\tau_2,s_n),
\eeq
where $p(z|\theta_2,\tau_2,s)$ is the density of the random variable $Z_k$ in the model
\eqref{eq:zmod}.
Recall that the singular values $s_n$ are known to the estimator.
Then, using these estimates, we perform the estimate of $\xbfhat_{2k}$ in line~\ref{line:x2}
of Algorithm~\ref{algo:em-vamp} using the estimated precision $\gamma_{2k} = 1/\tauhat_{2k}$.

To analyze this estimator rigorously,
we make the simplifying assumption that the random variable, $S$,
representing the distribution of singular values, is discrete.  That is, there are a finite number
of values $a_1,\cdots,a_L$ and probabilities $p_1,\ldots,p_L$ such that
\[
    \Pr(S^2=a_\ell) = p_\ell.
\]
Then, given vectors $\zbf_k$ and $\sbf$, define the $2L$ statistics
\[
    \phi_{\ell 0}(z,s) =  \indic{s^2 = a_\ell}, \quad
    \phi_{\ell 1}(z,s) = z^2\indic{s^2 = a_\ell}, \quad \ell=1,\ldots,L,
\]
and their empirical averages at time $k$,
\[
    \mu_{\ell 0} = \bkt{ \indic{s_n^2 = a_\ell} }, \quad
    \mu_{\ell 1} = \bkt{ z_{k,n}^2\indic{s_n^2 = a_\ell}}.
\]
Then, the ML estimate in \eqref{eq:theta2ML} can be rewritten  as
\beq \label{eq:theta2J}
    (\tauhat_{2k},\thetahat_{2k}) = \argmin_{\tau_2,\theta_2} J(\tau_2,\theta_2),
\eeq
where the objective function is
\beq \label{eq:theta2Jdis}
     J(\tau_2,\theta_2) :=
        = \frac{1}{N}\sum_{n=1}^N\left[ \frac{z_{k,n}^2}{s_n^2\tau_2 + \theta_2^{-1}} +
                \ln(s_n^2\tau_2 + \theta_2^{-1}) \right]
        = \sum_{\ell=1}^L \left[ \frac{\mu_{\ell 1}}{a_\ell\tau_2 + \theta_2^{-1}} +
                \mu_{\ell 0}\ln(a_\ell\tau_2 + \theta_2^{-1}) \right].
\eeq
Hence, we can compute the ML estimate \eqref{eq:theta2ML} from an objective function
computed from the $2L$ statistics, $\mu_{\ell 0}, \mu_{\ell 1}$, $\ell = 1,\ldots,L$.

To prove consistency of the ML estimate, first note that the SE analysis shows that
$\mu_{\ell 0} \arr \mubar_{\ell 0}$
and $\mu_{\ell 1} \arr \mubar_{\ell 1}$  where the limits are given by
\begin{align*}
    \mubar_{\ell 0} &= \Exp\left[ \phi_{\ell 0}(Z,S) \right] = P(S^2=a_\ell) = p_\ell, \\
    \mubar_{\ell 1} &= \Exp\left[ \phi_{\ell 1}(Z,S) \right] = \Exp(Z^2 | S^2=a_\ell)P(S^2=a_\ell) =
        p_\ell\left[ a_\ell \tau_{2k} + (\theta_2^0)^{-1} \right],
\end{align*}
where $(\tau_{2k},\theta_2^0)$ are the true values of $(\tau_2,\theta_2)$.
Hence, the objective function in \eqref{eq:theta2Jdis} will converge to
\beq \label{eq:theta2Jlim}
    \lim_{N \arr \infty}
        J(\tau_2,\theta_2) =
        \sum_{\ell=1}^L p_\ell \left[
            \frac{a_\ell \tau_2^0 + (\theta_2^0)^{-1}}{a_\ell\tau_2 + \theta_2^{-1}} +
                \mu_{\ell 0}\ln(a_\ell\tau_2 + \theta_2^{-1}) \right].
\eeq
Now, for any $y > 0$,
\[
    \argmin_{x > 0} \frac{y}{x} + \ln(x) = y.
\]
Thus, the minima of \eqref{eq:theta2Jlim} will occur at
\beq \label{eq:theta2lim}
    (\taubar_{2k},\thetabar_{2k}) := (\tau_{2k},\thetabf_1^0).
\eeq
It can also be verified that this minima is unique provided that there are at least
two different values $a_\ell$ with non-zero probabilities $p_\ell$.  That is, we require
that $S^2$ is not constant.  In this case, one can show that the mapping from the statistics
$\{ \mu_{\ell 0},\mu_{\ell 1})\}$ to the ML estimates
$(\tauhat_{2k},\thetahat_{2k})$ in \eqref{eq:theta2J} is
continuous at the limiting values $\{ \mubar_{\ell 0},\mubar_{\ell 1})\}$.
Hence, from the SE analysis, we obtain that $(\tauhat_{2k},\thetahat_{2k})
\arr (\taubar_{2k},\thetabar_{2k}) = (\tau_{2k},\theta_2^0)$,
in \eqref{eq:theta2lim}.   We have proven the following.

\begin{theorem} \label{thm:theta2cons}  Under the assumptions of Theorem~\ref{thm:em-se}, let $\theta_2^0$ denote the true noise precision and $\tau_{2k}$ be the true error variance
for some iteration $k$.
Suppose that
the random variable $S^2$ representing the squared singular values is discrete and non-constant.
Then, there exists a finite set of statistics $\phi_{\ell 0}(z_n,s_n)$ and $\phi_{\ell 1}(z_n,s_n)$
and parameter selection rule such that the estimates $(\tauhat_{2k},\thetahat_{2k})$
are asymptotically consistent in that
\[
    \lim_{N \arr \infty} (\tauhat_{2k},\thetahat_{2k})  = (\tau_{2k},\theta_1^0),
\]
almost surely.
\end{theorem}

\section{Simulation Details} \label{sec:sim}

\paragraph*{Sparse signal recovery}
To model the sparsity, $\xbf$ is drawn as an i.i.d.\ Bernoulli-Gaussian (i.e., spike and slab) prior,
\begin{equation}
    p(x_n|\thetabf_1)=(1-\beta_x)\delta(x_n)+\beta_x{\mathcal N}(x_n;\mu_x,\tau_x)
    \label{eq:BG},
\end{equation}
where parameters $\thetabf_1 = \{\beta_x,\mu_x,\tau_x\}$ represent the
sparsity rate $\beta_x\in(0,1]$, the active mean $\mu_x\in\R$, and the active variance $\tau_x>0$.
Following \cite{RanSchFle:14-ISIT,Vila:ICASSP:15},
we constructed $\Abf \in \R^{M \times N}$ from the SVD $\Abf=\Ubf\Sbf\Vbf\tran$, whose orthogonal matrices $\Ubf$ and $\Vbf$ were drawn uniformly with respect to the Haar measure and whose singular values $s_i$ were constructed as a geometric series, i.e., $s_i/s_{i-1}=\alpha~\forall i>1$, with $\alpha$ and $s_1$ chosen to achieve a desired condition number
$\kappa = s_1/s_{\min(M,N)}$ as well as $\|\Abf\|_F^2=N$.
This matrix model provides a good test for the stability of AMP methods since it is known that the matrix having
a high condition number is the main failure mechanism in AMP \cite{RanSchFle:14-ISIT}.
Recovery performance was assessed using normalized MSE (NMSE) $\|\xbfhat-\xbf\|^2/\|\xbf\|^2$ averaged over $100$ independent draws of $\Abf$, $\xbf$, and $\wbf$.

\begin{figure}
\centering
\includegraphics[width=0.75\columnwidth]{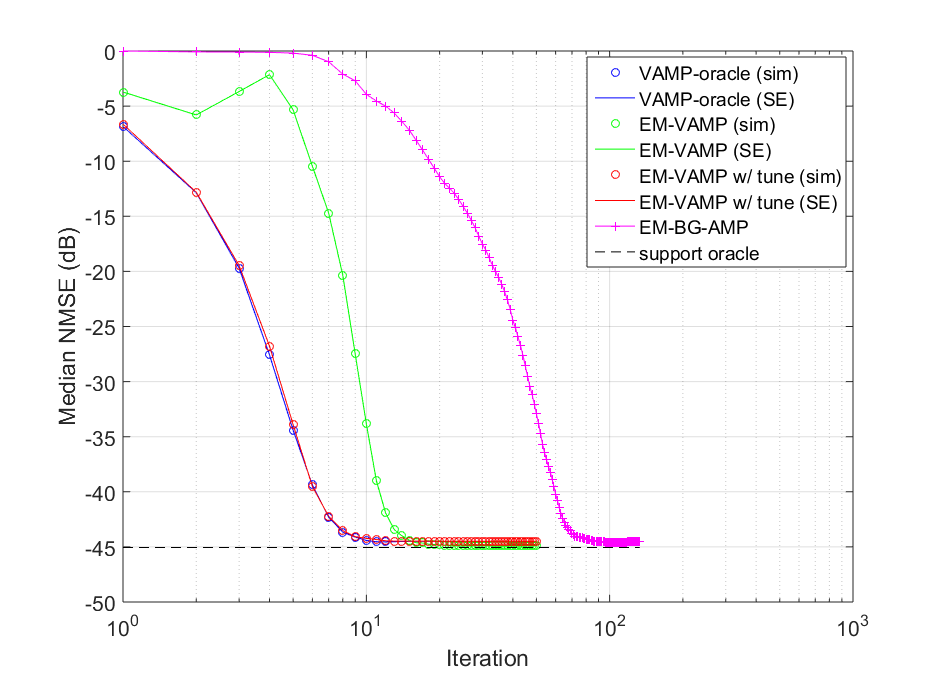}
\caption{Sparse signal recovery:  Plotted are the NMSE for various algorithms with a condition number of $\kappa = 10$. }
\label{fig:sparse_cond10}
\end{figure}

The left panel of figure~\ref{fig:sim} shows the NMSE versus iteration for various algorithms
under $M=512$, $N=1024$, $\beta_x=0.1$, $\mu_x=0$, and $(\tau_x,\theta_2)$ giving a signal-to-noise ratio of $40$~dB at a condition number of $\kappa=100$.
The identical simulation is shown in Fig.~\ref{fig:sparse_cond10}.  The trends in both figures are discussed in the main text.
The EM-VAMP algorithms were initialized with $\beta_x=(M/2)/N$, $\tau_x=\|\ybf\|^2 / \|\Abf\|_F^{2}\beta_x$, $\mu_x=0$, and $\theta_2^{-1}=M^{-1}\|\ybf\|^2$.
Four methods are compared: (i) VAMP-oracle, which is VAMP
under perfect knowledge of $\thetabf=\{\tau_w,\beta_x,\mu_x,\tau_x\}$;
(ii) EM-VAMP from \cite{fletcher2016emvamp};
(iii) EM-VAMP with auto-tuning described in Section~\ref{sec:autotune};
(iv)  the EM-BG-AMP algorithm from \cite{vila2013expectation} with damping from \cite{Vila:ICASSP:15}; and
(v) a lower bound obtain by an oracle estimator that knows the support of the vector $\xbf$.
For algorithms (i)--(iii), we have plotted both the simulated median NMSE as well as the state-evolution (SE)
predicted performance.

\paragraph*{Sparse image recovery}
The true image is shown in the top left panel of Fig.~\ref{fig:sat_images} and can be seen to be sparse in the pixel domain.
Three image recovery algorithms were tested:
(i) the EM-BG-AMP algorithm \cite{vila2013expectation} with damping from \cite{Vila:ICASSP:15};
(ii) basis pursuit denoising using SPGL1 as described in \cite{van2008probing}; and
(iii) EM-VAMP with auto-tuning described in Section~\ref{sec:autotune}.
For the EM-BG-AMP and EM-VAMP algorithms, the BG model was used, as before.
SPGL1 was run in ``BPDN mode," which solves $\min_{\xbf} \|\xbf\|_1$ subject to $\|\ybf-\Abf\xbf\| \leq \sigma$ for true noise variance $\sigma^2$.
All algorithms are run over 51 trials, with the median value plotted in Fig.~\ref{fig:sim}.
While the EM-VAMP theory provides convergence guarantees when the true distribution fits the
model class and in the large-system limit, we found some damping was required on finite-size, real images that are not generated from the assumed
prior.  In this experiment,
to obtain stable convergence, the EM-VAMP algorithm was aggressively damped as described in
\cite{rangan2016vamp}, using a damping ratio of 0.5 for the first- and second-order terms.

The computation times as a function of condition number $\kappa$ are shown in Fig.~\ref{fig:sat_comp_time}.
All three methods benefit from the fact that the matrix $\Abf$ has a fast implementation.  In addition,
for EM-VAMP, the matrix has an easily computable SVD\@.  The figure shows that the EM-VAMP algorithm achieves dramatically faster performance at high condition numbers -- approximately 10 times faster than either EM-BG-AMP or SPGL1.
The various panels of Fig.~\ref{fig:sat_images} show the images recovered by each algorithm at condition number $\kappa=100$ for the realization that achieved the median value in Fig.~\ref{fig:sim} (right panel).

\begin{figure}
\centering
\includegraphics[width=0.8\columnwidth]{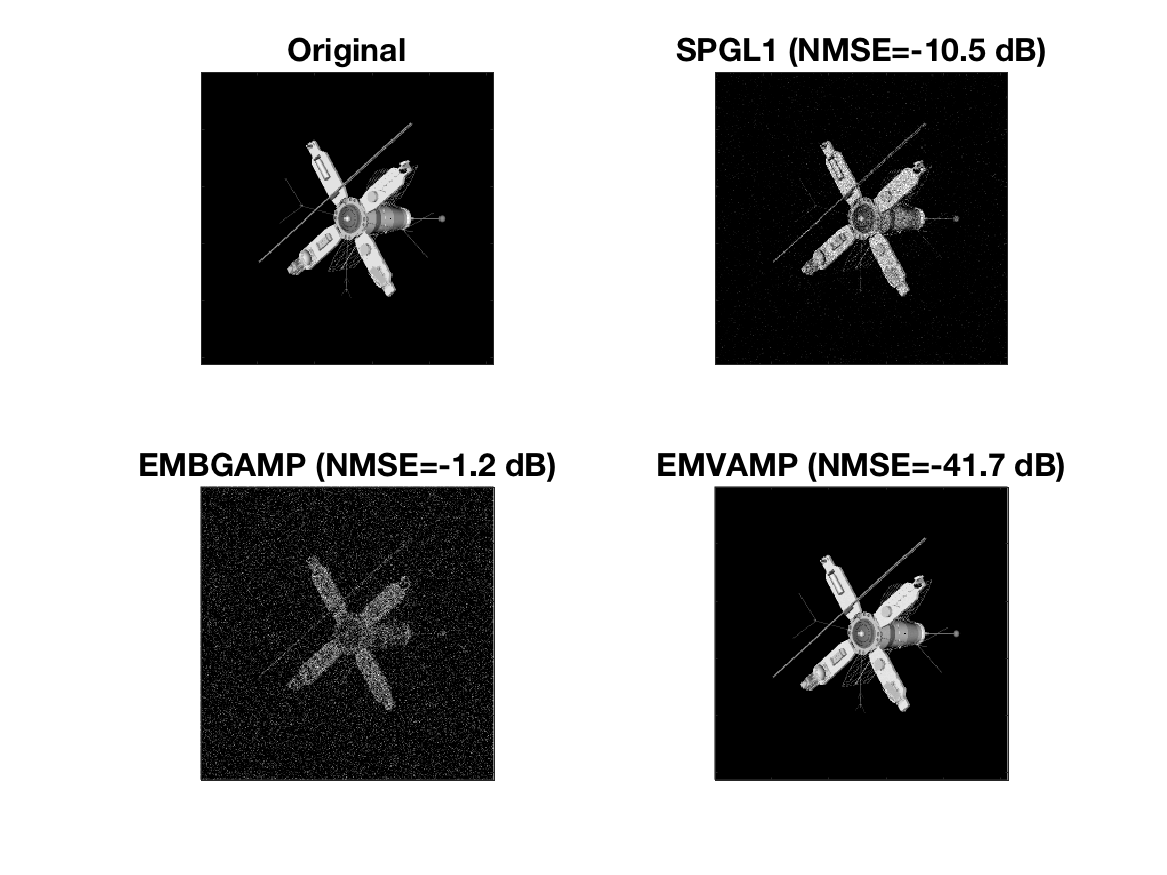}
\caption{Sparse image recovery:  Original 'satellite' image as used in \cite{Vila:ICASSP:15} along
with the median-NMSE recovered images for the various algorithms at measurement ratio
$M/N = 0.5$, condition number $\kappa=100$, and 40~dB SNR\@.
}
\label{fig:sat_images}
\end{figure}

\begin{figure}
\centering
\includegraphics[width=0.65\columnwidth]{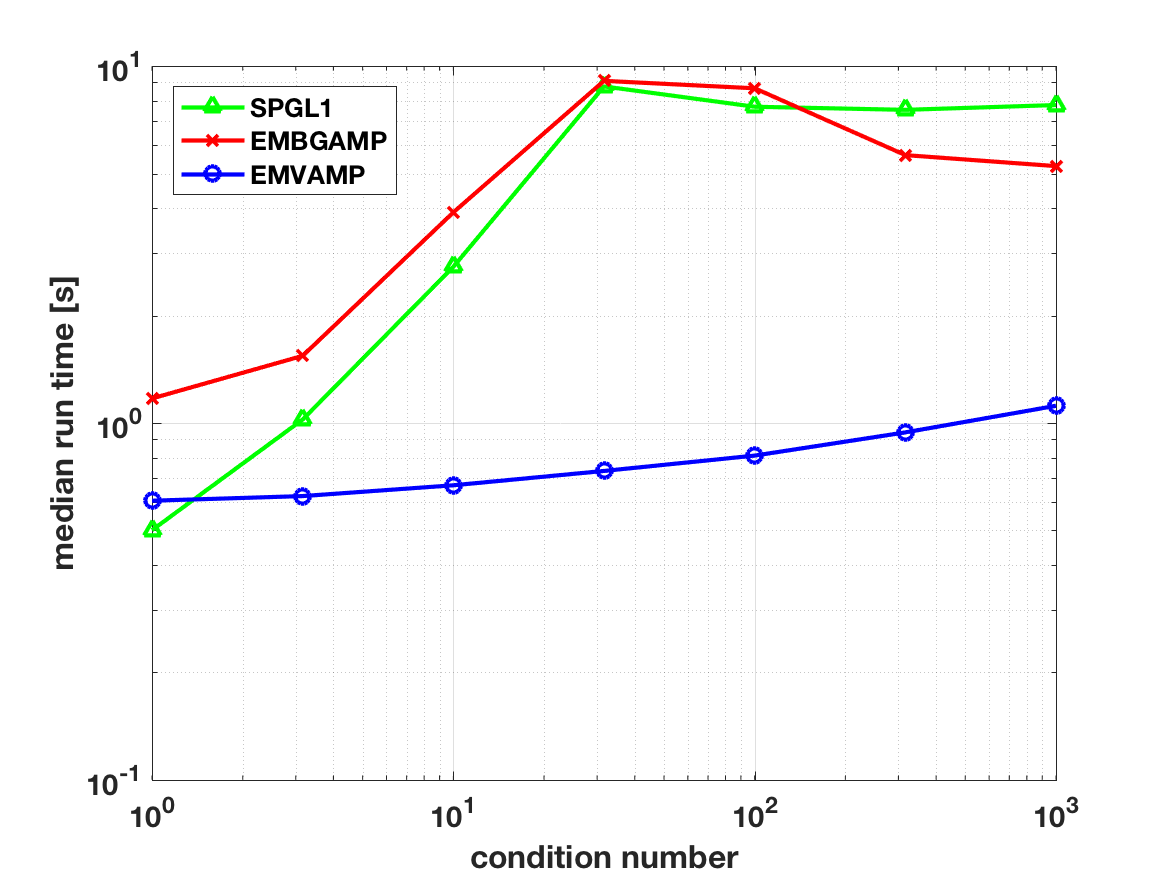}
\caption{Computation time versus condition number for various algorithms in the satellite image recovery problem. }
\label{fig:sat_comp_time}
\end{figure}

\newcommand{\SortNoop}[1]{}

\end{document}